\documentclass[%
 reprint,
nofootinbib,
amsmath,
amssymb,
aps,
showkeys,
onecolumn
]{revtex4-2}

\makeatletter
\renewcommand{\fnum@figure}{\textbf{Fig.} \thefigure}
\makeatother
\renewcommand*{\thefigure}{\textbf{\arabic{figure}}}

\usepackage[english]{babel}
\usepackage{amsthm}
\usepackage{pgfplots}
\pgfplotsset{compat=1.18}
\usepackage{amsmath}
\usepackage{graphicx}% Include figure files
\usepackage{dcolumn}% Align table columns on decimal point
\usepackage{bm}% bold math
\usepackage{txfonts}
\usepackage{mathtools}
\usepackage[colorlinks=true, allcolors=blue]{hyperref}
\usepackage{esint}
\usepackage{siunitx}
\usepackage{float}
\usepackage{framed}
\usepackage{comment}
\usepackage{wasysym}
\usepackage{threeparttable}
\usepackage{tabularx}
\usepackage{amsfonts} 
\usepackage{amssymb}   
\usepackage{placeins}

\usepackage{tabularx}
\usepackage{seqsplit}
\usepackage{booktabs}
\usepackage[ruled,vlined,linesnumbered]{algorithm2e} 
\usepackage{setspace}                                
\usepackage{xcolor}
\usepackage{tikz}       
\usetikzlibrary{trees, arrows.meta}
\usepackage{caption}

\newcommand\mycommfont[1]{\footnotesize\ttfamily\color{gray}}
\SetCommentSty{mycommfont}
\SetKwInput{Input}{Input}     
\SetKwInput{Output}{Output}   
\usepackage{orcidlink}
\usepackage{cancel}
\usepackage{makecell}
\usepackage{float}
\usepackage{mathrsfs}
\usepackage{multirow}
\usepackage[table]{xcolor}
\usepackage{siunitx}
\usepackage{tikz}
\usetikzlibrary{automata, positioning, arrows.meta}

\newtheorem{theorem}{Theorem}[section]
\newtheorem{corollary}{Corollary}[theorem]

\newtheorem{lemma}{Lemma}[section]
\newtheorem{definition}{Definition}[section]

\newtheorem{observation}[theorem]{Observation}

\newcommand\mydots{\hbox to 1em{.\hss.\hss.}}

\begin{document}
\selectlanguage{english}

\preprint{APS/123-QED}

\begin{comment}
\end{comment}

\author{Wawrzyniec Bieniawski\textsuperscript{*} \protect\orcidlink{0009-0000-9943-886X}}
%\affiliation{Łukaszyk Patent Attorneys, ul. Głowackiego 8, 40-052 Katowice, Poland}
\affiliation{Dom Zatorski Scientific Foundation,
%ul.~Odległa 10, 41-050 
Chorzów, Poland}
\thanks{\href{mailto:wbieniawski@patent.pl}{wbieniawski@patent.pl}}
%ORCID 0009-0000-9943-886X

\hypersetup{
%colorlinks=true,
%linkcolor=red,
%filecolor=magenta,      
%urlcolor=cyan,
%pdfpagemode=FullScreen,
%citecolor=green,
pdftitle={@ASI vs OTHER ALGOS@},
bookmarksopen=true,
}

\title{Assembly Theory and the Smallest Grammar Problem}

\begin{abstract}

Assembly theory (AT) quantifies complexity through the assembly index (ASI) --- a metric proven equivalent to the size of the smallest straight-line program (SLP). While this equivalence links AT to data compression—a field shaped by decades of research—the practical efficacy of specific algorithms as ASI approximators remains largely unexplored. This study provides an empirical evaluation of eight compression algorithms, spanning grammar-based and dictionary schemes (CAs) across a dataset of 408 strings --- comprising 368 synthetic strings (including max-complexity strings and strings with varying levels of symbol distribution balance, quantified by Shannon entropy) and 40 natural biological sequences (genomic and proteomic). We introduce Re-Pair T-NDR, a branch-and-bound tie-resolving variant of Re-Pair, providing a tighter upper bound on the ASI than the other CAs researched in this study. Increasing alphabet size extends the regime in which CAs closely track the ASI, delaying a "divergence point", where a CA deviates from the optimal assembly path. The results establish compression algorithms as practical tools for ASI estimation.

\end{abstract}


\maketitle

\section{Introduction}\label{sec:Introduction}

The assembly index (ASI) measures the minimal number of assembly steps required to construct an object from elementary components, providing a quantifiable metric of constructive complexity~\cite{marshall_probabilistic_2017, walker_probabilistic_2019, meadows_planetary_2020, liu_exploring_2021, marshall_identifying_2021, marshall_formalising_2022, sharma_assembly_2023, jirasek_investigating_2024, lukaszyk_assembly_2024, raubitzek_autocatalytic_2024, flamm_assembly_2025}. Recent theoretical work has established the ASI as equivalent to the size of the smallest straight-line program (SLP) generating a target string~\cite{masierak_computational_2026}, thereby connecting Assembly Theory (AT) to the smallest grammar problem (SGP)~\cite{charikar_smallest_2005}. This equivalence is computationally significant: 
while the SGP is NP-hard~\cite{charikar_smallest_2005} and the ASI decision problem is NP-complete~\cite{masierak_computational_2026, lukaszyk_optimal_2026}, both are computable, unlike Kolmogorov complexity.

Despite this theoretical foundation, a systematic empirical evaluation of which state-of-the-art compressors provide reliable ASI estimates in practice remains absent from the literature. This gap leaves open the question of how faithfully existing algorithms approximate the ASI, particularly for real-world sequences where exact computation is infeasible. We address this through two questions:

\begin{enumerate}
\item Do state-of-the-art compressors provide sufficient heuristic efficacy for ASI estimation, and what are their error bounds?
\item How do alphabet size and symbol distribution balance (quantified by Shannon entropy~$H$) determine algorithmic accuracy, and can selection guidelines be established?
\end{enumerate}

We evaluate eight grammar-based and dictionary compression algorithms across 408 strings---368 synthetic (including max-complexity strings and entropy-stratified samples) and 40 biological sequences (proteomic and genomic). Our principal contributions are:

\begin{enumerate}
\item{Re-Pair T-NDR:} A branch-and-bound variant of Re-Pair that exhaustively resolves tie-breaking decisions, empirically dominating both deterministic variants on all tested strings (Observation~\ref{obs:tndr_empirical}).

\item{Alphabet-size effect:} We formalise the \emph{divergence point}~$N^*(A, |\Sigma|)$---the smallest length at which an algorithm can deviate from the optimal ASI---and demonstrate that larger alphabets delay this threshold.

\item{Rank-correlation artefacts:} Partial Spearman analysis reveals that high rank correlations on max-complexity sets are partly driven by length covariation. Genuine length-independent complexity discrimination is established on fixed-length sets ($N=30$), where Re-Pair T-NDR achieves $\rho \geq 0.916$ (binary), $\rho \geq 0.987$ (ternary), and $\rho \geq 0.987$ (quaternary).
\end{enumerate}

Throughout this paper, the terms \textit{heuristic} and \textit{algorithm} are used interchangeably. Regarding notation, we use $N$ to denote string length throughout this paper, although much of the cited literature on the SGP and addition chains traditionally employs $n$ (a convention we exceptionally retain in Section~\ref{sec:Theory and Background} and in Section~\ref{sec:Discussion} to maintain alignment with the original formulations).

The paper is organised as follows. Section~\ref{sec:Theory and Background} reviews grammar-based compression and formalises the ASI--SLP equivalence. Section~\ref{sec:Methods} describes the dataset, algorithms, and evaluation metrics. Section~\ref{sec:Results} presents empirical results across synthetic and biological data. Section~\ref{sec:Discussion} situates the findings within AT and responds to recent criticisms. Section~\ref{sec:Conclusions} summarises contributions.

\section{Theory \& Background}\label{sec:Theory and Background}

Grammar-based compression algorithms (GBCA) seek a context-free grammar (CFG) that generates exactly one target string while minimising the total number of rules and symbols. Unlike dictionary methods (LZ family), which exploit local redundancies through single-pass sliding windows, GBCA capture hierarchical, self-similar structure via recursive rewrite rules. Early practical algorithms include Sequitur (1997) and Re-Pair (1999)~\cite{nevill-manning_identifying_1997, larsson_offline_1999}; subsequent theoretical work established approximation ratios ranging from $O((n/\log n)^{2/3})$ for greedy heuristics to $O(\log n/\log\log n)$ for methods based on Yao's addition-chain framework~\cite{charikar_smallest_2005, yao_evaluation_1976}. Table~\ref{tab:approx_ratio} summarises these bounds and indicates which algorithms were selected for empirical evaluation in this study.

\subsection{The ASI--SLP Equivalence}\label{subsec:asi_slp}

An SLP is a CFG in which each nonterminal has exactly one production rule and the dependency graph is acyclic, ensuring that it generates a single, unique string~\cite{ganardi_compression_2021}. 
The size of an SLP is the number of its production rules, each representing a binary concatenation of two symbols (nonterminals or terminals). Under this convention, the SLP size equals the number of assembly steps, and the SGP asks for the minimum-size SLP generating a given string~$w$.

The connection between the ASI of a string and the size of a corresponding CFG was initiated in~\cite{abrahao_assembly_2024}. This was later refined into an exact equivalence: the ASI of a string equals the size of its smallest SLP~\cite{masierak_computational_2026}:
\begin{equation}\label{eq:asi_slp}
    \mathrm{ASI}(w) = \mathrm{SLP}(w).
\end{equation}

Equation~\eqref{eq:asi_slp} places the ASI decision problem in the NP-complete class~\cite{masierak_computational_2026, lukaszyk_optimal_2026}, and consequently establishes the ASI as the theoretical lower bound---the ``ideal compressor''---within the grammar-based compression framework.

\subsection{Re-Pair and the Tie-Breaking Problem}\label{subsec:repair_ties}

Re-Pair operates by iteratively replacing the most frequent digram with a nonterminal until no digram occurs more than once~\cite{larsson_offline_1999}. When multiple digrams share the maximum frequency, a tie-breaking rule must be applied. Two standard deterministic variants are:
\begin{itemize}
    \item {Re-Pair FIFO:} selects the first-encountered digram during left-to-right scanning;
    \item {Re-Pair FIXED:} selects the lexicographically smallest digram.
\end{itemize}
Neither variant is guaranteed to yield the smallest grammar: a suboptimal choice at one iteration can propagate through the hierarchy and compound over subsequent steps. This motivates an exhaustive search over all tied digrams at each iteration.

\subsection{Re-Pair T-NDR}\label{subsec:tndr_theorem}

Let $\mathcal{P}(w)$ denote the set of all Re-Pair execution paths on a string $w$, where each path corresponds to a sequence of tie-breaking choices. For any path $p \in \mathcal{P}(w)$, let $L_p(w)$ denote the grammar size produced by that path.

\begin{lemma}[Lower bound on grammar size]
\label{lem:lowerbound}
Let $w'$ be the string remaining after $r$ Re-Pair substitution rounds have been applied to $w$,
and let $N=|w'|$. Any straight-line program completing this partial derivation has total size at least
$r+\lceil \log_2 N\rceil$.
\end{lemma}

\begin{proof}
Consider any completion of the partial derivation into an SLP.
Let $c$ be the total number of additional binary concatenations required: the number of new rules plus the $(m-1)$ concatenations needed to assemble the $m$ symbols that remain once no digram repeats into a single root. Assign weight $1$ to each of the $N$ symbols of $w'$. Each binary concatenation creates a symbol whose weight is the sum of the weights of its two children, so the maximum weight among all symbols at most doubles in a single step. After $c$ binary concatenations the maximum weight is therefore at most $2^{c}$. Since the root symbol must derive the entire string $w'$, its weight must equal $N$. Hence $2^{c} \ge N$, i.e.\ $c\ge\lceil\log_{2}N\rceil$. Adding the $r$ rounds already spent yields total size at least $r+\lceil\log_{2}N\rceil$.
\end{proof}

\begin{theorem}[Correctness and Soundness of Re-Pair T-NDR]
\label{thm:tndr}
For any string $w \in \Sigma^+$:
\begin{enumerate}
\item[(a)] Re-Pair T-NDR enumerates all paths in $\mathcal{P}(w)$ that are not pruned by the branch-and-bound condition or short-circuited by a cache hit.
\item[(b)] If a partial path with current rule count $r$ and current string length $N$ satisfies $r+\lceil \log_2 N\rceil \ge G_{\text{best}}$, where $G_{\text{best}}$ is the best complete grammar size found so far, then no completion of this partial path can improve upon $G_{\text{best}}$.
\item[(c)] Re-Pair T-NDR returns $L_{\text{T-NDR}}(w) \ge \min_{p \in \mathcal{P}(w)} L_p(w)$; in particular, $\mathrm{ASI}(w) \le L_{\text{T-NDR}}(w)$.
\end{enumerate}
\end{theorem}

\begin{proof}
\textbf{(a)} The search tree branches on all digrams sharing maximum frequency at each iteration; each branch corresponds to exactly one path in $\mathcal{P}(w)$. Paths are omitted only when pruned by the branch-and-bound condition or short-circuited by a cache hit.

\textbf{(b)} By Lemma~\ref{lem:lowerbound}, any completion of a partial path requires at least $r + \lceil \log_2 N \rceil$ rules total. If this lower bound meets or exceeds $G_{\text{best}}$, no completion can improve the incumbent.

\textbf{(c)} The tree is finite, since each substitution strictly reduces string length. By (a), all paths that are neither pruned nor short-circuited are explored; by (b), pruning never discards a path that could yield a better solution than the incumbent at the time of pruning. The algorithm returns the minimum over all explored paths. Memoisation may reuse a cached result computed under a tighter global bound, potentially omitting branches that would become viable under a looser bound; hence, the returned value may overestimate but never underestimate the true optimum: $L_{\text{T-NDR}}(w) \ge \min_{p\in\mathcal{P}(w)} L_p(w)$.

\end{proof}

\begin{corollary}[ASI Upper Bound]\label{cor:tndr_bound}
For all $w \in \Sigma^+$, $\mathrm{ASI}(w) \leq L_{\text{T-NDR}}(w)$. Equality holds only if the optimal SLP is reachable through some sequence of Re-Pair substitutions and the branch-and-bound search with memoisation discovers it without being blocked by a cached overestimate.
\end{corollary}

\begin{proof}
By equation~\eqref{eq:asi_slp}, $\mathrm{ASI}(w) = \min_{\text{all SLPs}} |\text{SLP}| 
\le \min_{p \in \mathcal{P}(w)} L_p(w) \le L_{\text{T-NDR}}(w)$, 
where the last inequality is Theorem~\ref{thm:tndr}(c). 
If equality $\mathrm{ASI}(w) = L_{\text{T-NDR}}(w)$ holds, then all inequalities above must be equalities. In particular, $\min_{p\in\mathcal{P}(w)} L_p(w) = \mathrm{ASI}(w)$, so the optimal SLP is reachable through some Re-Pair path; and $L_{\text{T-NDR}}(w) = \min_{p\in\mathcal{P}(w)} L_p(w)$, so the search returned the true minimum without being blocked by a cached overestimate.
\end{proof}

\begin{observation}[Empirical dominance of T-NDR over FIFO and FIXED]\label{obs:tndr_empirical}
Exhaustive execution on all $408$ strings confirmed that T-NDR never yields a larger grammar than FIFO or FIXED. Excluding the $83$ unary strings where all variants coincide, T-NDR produced a strictly smaller grammar than both deterministic variants in $31$ of the remaining $325$ strings (improvements of $1$--$3$ steps relative to $\min(\text{FIFO},\text{FIXED})$), with zero instances of inferior performance.
\end{observation}

\subsection{Scope and Limitations}\label{subsec:scope}

Both the ASI and GBCA operate on linear string representations without modelling geometric or thermodynamic constraints. Consequently, the ASI quantifies informational and constructional complexity; minimal grammatical descriptions need not correspond to energetically minimal or experimentally feasible synthesis pathways. Results should be interpreted as statements about algorithmic structure and compressibility, not as direct models of physical assembly.

\begin{table}[H]
\centering
\caption{Approximation ratio comparison of selected algorithms. Algorithms marked with \checkmark\ in the ``Evaluated'' column were included in the empirical study of Section~\ref{sec:Results}.}
\label{tab:approx_ratio}
\small
\begin{tabular}{|l|c|c|c|c|}
\hline
\textbf{Algorithm} & \textbf{Upper Bound} & \textbf{Lower Bound} & \textbf{Reference} & \textbf{Evaluated} \\
\hline
Limit (Yao based) & $O(\log n / \log \log n)$ & --- & \cite{yao_evaluation_1976} & $\times$ \\
\hline
Charikar/Rytter/Sakamoto & $O(\log n)$ & --- & \cite{rytter_application_2003, charikar_smallest_2005, sakamoto_space-saving_2009} & $\times$ \\
\hline
Charikar/Rytter/Sakamoto$^*$ & $O(\log(n/g^*))$ & --- & \cite{rytter_application_2003, charikar_smallest_2005, sakamoto_fully_2005} & \checkmark \\
\hline
LCA & $O(\log^2 n)$ & --- & \cite{maruyama_fully-online_2013, takabatake_space-optimal_2017} & \checkmark\ (SOLCA) \\
\hline
Charikar VIIA & $O(\log^3 n)$ & --- & \cite{charikar_approximating_2002} & $\times$ \\
\hline
Bisection & $O((n/\log n)^{1/2})$ & $\Omega((n/\log n)^{1/2})$ & \cite{kieffer_universal_2000, charikar_smallest_2005} & \checkmark \\
\hline
LZ78 family & $O((n/\log n)^{2/3})$ & $\Omega(n^{2/3}/{\log n})$ & \cite{charikar_smallest_2005, sakamoto_space-saving_2009} & \checkmark\ (LZD) \\
\hline
Re-Pair & $O((n/\log n)^{2/3})$ & $\Omega(\log n/ \log \log n)$ & \cite{Lehman2002Approximation, charikar_smallest_2005} & \checkmark\ (FIFO, FIXED, T-NDR) \\
\hline
Sequitur/Sequential & $O((n/\log n)^{3/4})$ & $\Omega(n^{1/3})$ & \cite{nevill-manning_identifying_1997, kieffer_universal_2000, charikar_smallest_2005} & \checkmark\ (Sequitur) \\
\hline
\end{tabular}
\begin{tablenotes}
\small
\item[$*$] $g^*$ denotes the optimal grammar size.
\item Note: asymptotic ordering; for specific $n$, numerical values may differ from growth rate ordering.
\item Note: $\log n$ refers to $\log_2(n)$.
\end{tablenotes}
\end{table}

\section{Methods}\label{sec:Methods}

\subsection{Dataset}\label{sec:dataset}

The synthetic dataset we used for numerical experiments comprises: 

\begin{enumerate}
\item 83 unary strings ($\Sigma = 1$, $3 \leq N \leq 85$, not shown in Table Section);
\item 79 maxASI binary strings ($\Sigma = 2$, $7 \leq N \leq 85$, Table~\ref{Table:maxbinarya}, Table~\ref{Table:maxbinaryb});
\item 42 maxASI ternary strings ($\Sigma = 3$, $13 \leq N \leq 54$, Table~\ref{Table:maxternary});
\item 29 maxASI quaternary strings ($\Sigma = 4$, $21 \leq N \leq 49$, Table~\ref{Table:maxquaternary});
\item 20 random binary strings ($\Sigma = 2$, $41 \le N \le 71$, Table~\ref{Table:random_b2.pdf}, including two minASI strings for $N \in \{64, 68\}$);
\item 25 random fixed-length binary strings ($\Sigma = 2$, $N=30$, Table~\ref{Table:fixed_b2.pdf}, selected across five Shannon entropy intervals $0.21 < H < 0.997$);
\item 20 random ternary strings ($\Sigma = 3$, $20 \le N \le 51$, Table~\ref{Table:random_b3.pdf}, including a  minASI string of $N=27$);
\item 25 random fixed-length ternary strings ($\Sigma = 3$, $N = 30$, Table~\ref{Table:fixed_b3.pdf}, selected across five Shannon entropy intervals $0.41 < H < 1.57$); 
\item 20 random quaternary strings ($\Sigma = 4$, $20 \le N \le 50$, Table~\ref{Table:random_b4.pdf}, including a minASI string of $N=32$);
\item 25 random fixed-length quaternary strings ($\Sigma = 4$, $N = 30$, Table~\ref{Table:fixed_b4.pdf}, selected across five Shannon entropy intervals $0.62 < H < 1.99$).
\end{enumerate}

Additionally, 40 biological sequences were evaluated: 20 proteins ($|\Sigma| \approx 20$) from the RCSB Protein Data Bank and 20 nucleic acid sequences ($|\Sigma|=4$) from NCBI GenBank (Table~\ref{tab:biological_results_all}).

MaxASI strings are nearly balanced strings of maximal grammatical complexity for their length, selected from our previous study~\cite{bieniawski_assembly_2026}. Their ASI values were computed via exhaustive search with the admissible lower bound $r+\lceil\log_2 N\rceil$ from Lemma~\ref{lem:lowerbound}; because this bound is a true mathematical lower bound on the size of any SLP completing the partial derivation, pruning by it never discards a branch containing an optimal solution. The ASI finder employs this bound exclusively and does not use the lossy canonical-form memoisation of Algorithm~\ref{alg:t_ndr_final_fixed}; consequently, the ground-truth values are not subject to the overestimation caveat of Corollary~\ref{cor:tndr_bound}. For unary strings $a^N$, the ASI equals the length of the shortest addition chain for $N$ (OEIS A003313)~\cite{A003313}; our computed values coincide with this sequence for every $3 \le N \le 85$. Full implementation details of the ASI finder are available in the \texttt{szluk/AssemblyTheory} repository (see Data Availability Statement).

Shannon entropy for a string $w$ with symbol frequencies $p_k$ is defined as
\begin{equation}\label{eq:shannon_entropy}
    H(w) = -\sum_{k=1}^{|\Sigma|} p_k \log_2 p_k.
\end{equation}

$H(w)$ measures symbol distribution balance, not complexity per se: two perfectly balanced strings can have different ASI values (e.g., $\mathrm{ASI}[01010101]=3$ vs. $\mathrm{ASI}[00011101]=6$, both with $H=1$).

\subsection{Algorithms}\label{sec:algorithms}

Table~\ref{tab:algorithm_summary} lists the eight evaluated algorithms. All implementations were modified to output the algorithm-specific index $L_A(w)$ defined as:
\begin{definition}[the algorithm-specific index]
An algorithm-specific index $L_A(w)$ is the number of binary concatenations in the grammar produced by $A$ to compress the string $w$, excluding terminal symbols from the initial count. 
\end{definition}
$L_A(w)$ matches the ASI definition, where each assembly step represents a binary concatenation.
Thus, to make a compression algorithm $A$ comparable to ASI($w$), we modified all implementations of the evaluated algorithms listed in Table~\ref{tab:algorithm_summary} to output $L_A(w)$ directly. 

\begin{table}[H]
\centering
\caption{Summary of evaluated algorithms. Mode: O = online, 
F = offline. All implementations are publicly available; see Data Availability Statement for repository links and pseudocode references.}
\label{tab:algorithm_summary}
\small
\begin{threeparttable}
\begin{tabular}{|l|c|l|l|}
\hline
\textbf{Algorithm} & \textbf{Mode} & \textbf{Strategy} 
    & \textbf{Ref.} \\
\hline
Sakamoto      & F & Run decomp.\ + greedy digram           
    & \cite{sakamoto_fully_2005} \\
SOLCA         & O & ESP + LCA on binary parse trees         
    & \cite{takabatake_space-optimal_2017} \\
Bisection     & F & Recursive power-of-two split            
    & \cite{kieffer_universal_2000} \\
Sequitur      & O & Digram uniqueness + rule utility        
    & \cite{nevill-manning_identifying_1997} \\
Re-Pair FIFO  & F & Max-freq digram; insertion-order ties   
    & \cite{larsson_offline_1999} \\
Re-Pair FIXED & F & Max-freq digram; lexicographic ties     
    & \cite{larsson_offline_1999} \\
Re-Pair T-NDR & F & Branch-and-bound over all tied digrams  
    & \cite{larsson_offline_1999} \\
LZD           & F & LZ78 factor-pair concatenation        
    & \cite{10.1007/978-3-319-19929-0_19} \\
\hline
\end{tabular}
\end{threeparttable}
\end{table}

FIFO and FIXED are deterministic tie-breaking rules (Section~\ref{subsec:repair_ties}). T-NDR replaces fixed tie-breaking with exhaustive branch-and-bound search over all tied digrams, using the lower bound $S_{\mathrm{lower}} = r + \lceil \log_2 N \rceil$ for pruning and canonical-form memoisation (Theorem~\ref{thm:tndr}). In the worst case the branching factor is exponential; in practice, pruning and memoisation keep computation tractable for $N \lesssim 300$ (for selected biological strings $|\Sigma| \approx 20$, $|\Sigma|=4$) on commodity hardware. A detailed empirical analysis of search-space size and runtime scaling is deferred to future work.

\subsection{Hardware and Implementation}\label{sec:Hardware}

Heuristic calculations were performed on a commodity desktop (Intel Celeron G1840, 2.80~GHz, 8~GB DDR3 RAM). ASI computations used an AMD Ryzen 9 3950X (16 cores, 3.5~GHz, 128~GB DDR4) and, for selected instances, the LUMI supercomputer (EuroHPC). All timing measurements for the Pareto analysis were conducted on the commodity desktop to ensure comparability.

\subsection{Evaluation Metrics}\label{sec:metrics}

The mean relative error quantifies approximation quality:
\begin{equation}\label{eq:mean_relative_error}
    \overline{\Delta_A} = \frac{1}{m} \sum_{i=1}^{m} \frac{L_A(w_i) - \mathrm{ASI}(w_i)}{\mathrm{ASI}(w_i)}.
\end{equation}
$\overline{\Delta_A} = 0$ indicates exact agreement; positive values measure overestimation.

Spearman's $\rho_A$ and Kendall's $\tau_A$ rank correlation assess whether an algorithm preserves the complexity ranking of strings. A high $\rho_A$ with moderate $\overline{\Delta_A}$ indicates that the algorithm is useful for comparative studies despite absolute bias.

The entropy--error correlation coefficient $\rho_{(H(w),\Delta_A)}$ measures whether error grows with Shannon entropy. Positive values indicate degradation in balanced-alphabet regimes; near-zero values indicate entropy-invariant performance.

Pareto efficiency - an algorithm $A$ dominates $B$ ($A \succ B$) if it is strictly better in at least one objective (error or time) and no worse in the other:
\begin{equation}\label{eq:dominance}
    (\overline{\Delta_A} \leq \overline{\Delta_B} \land C_A < C_B) \lor (\overline{\Delta_A} < \overline{\Delta_B} \land C_A \leq C_B).
\end{equation}
The Pareto frontier $\mathcal{P}$ comprises all tested algorithms (Figure~\ref{fig:pareto}).

Pairwise differences in mean relative error were assessed via Wilcoxon signed-rank tests on set-level means ($m=10$ paired observations), with Holm--Bonferroni correction and rank-biserial effect sizes.

Significance stars in Tables~\ref{tab:corr_binary_entropy}--\ref{tab:corr_quaternary_entropyn30} 
report uncorrected $p$-values for individual Spearman and Kendall tests. These are exploratory correlations rather than confirmatory hypotheses; we do not apply family-wise error correction because (i) each cell reports a single pre-specified metric for one algorithm--dataset pair, and (ii) the primary comparative inference is supported by the Wilcoxon signed-rank analysis in Table~\ref{tab:wilcoxon}, which uses Holm--Bonferroni correction. 
Readers should interpret the starred correlations as descriptive strength indicators, not as simultaneous guarantees.

\section{Results}\label{sec:Results}

This section presents the results of the numerical experiments on the dataset from subsection~\ref{sec:dataset} using the algorithms described in subsection~\ref{sec:algorithms}.

\subsection{MaxASI Strings}

\subsubsection{Binary, Ternary, and Quaternary maxASI}

\begin{figure}[htbp]
    \centering
    \includegraphics[width=0.9\textwidth]{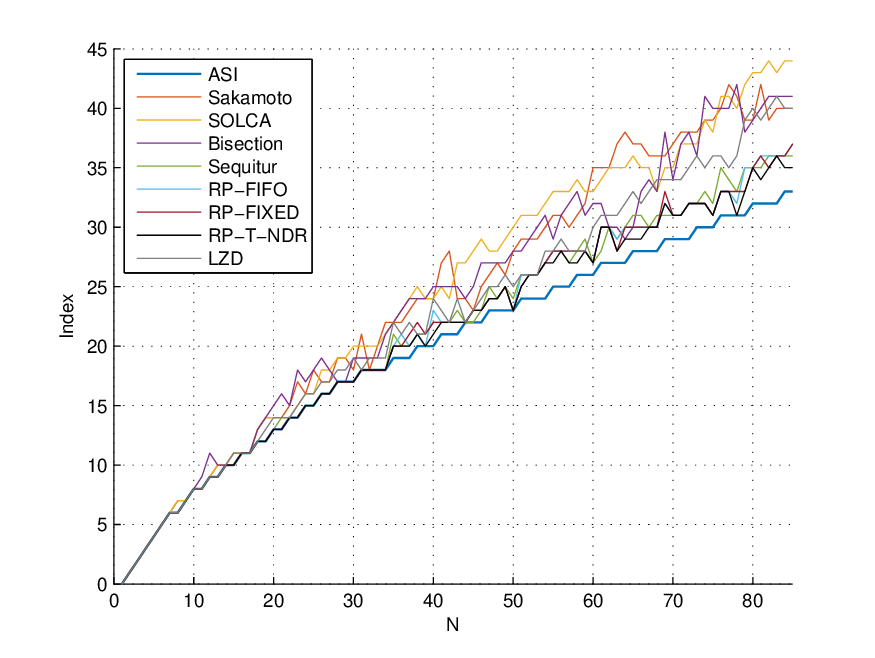}
    \caption{ASI vs.\ algorithm-specific index for the binary maxASI strings.}
    \label{fig:Fig_b2}
\end{figure}

\begin{table}[htbp]
\centering
\small
\caption{Rank correlations and mean relative error for the binary maxASI strings.}
\label{tab:corr_binary_maxasi}
\begin{tabular}{|l|c|c|c|}
\hline
\textbf{Algorithm $A$} &
\textbf{Spearman's $\rho_A$} &
\textbf{Kendall's $\tau_A$} &
\textbf{Mean error $\overline{\Delta_A}$ [\%]} \\
\hline
Sakamoto      & 0.9919 & 0.9507 & 18.52 \\
SOLCA         & 0.9963 & 0.9753 & 20.41 \\
Bisection     & 0.9883 & 0.9359 & 17.28 \\
Sequitur      & 0.9977 & 0.9822 &  5.34 \\
Re-Pair FIFO  & 0.9973 & 0.9801 &  4.96 \\
Re-Pair FIXED & 0.9974 & 0.9816 &  4.89 \\
Re-Pair T-NDR & 0.9968 & 0.9773 &  4.03 \\
LZD           & 0.9963 & 0.9717 & 10.75 \\
\hline
\end{tabular}
\end{table}

All tested algorithms demonstrate high rank correlations on the binary maxASI set ($79$ strings, $7 \leq N \leq 85$). Sequitur achieves the highest $\rho$, while Re-Pair T-NDR attains the lowest mean error ($4.03\%$). SOLCA and Sakamoto, despite strong rank preservation, exhibit the highest mean errors ($\sim$20\%).

\begin{figure}[htbp]
    \centering
    \includegraphics[width=0.9\textwidth]{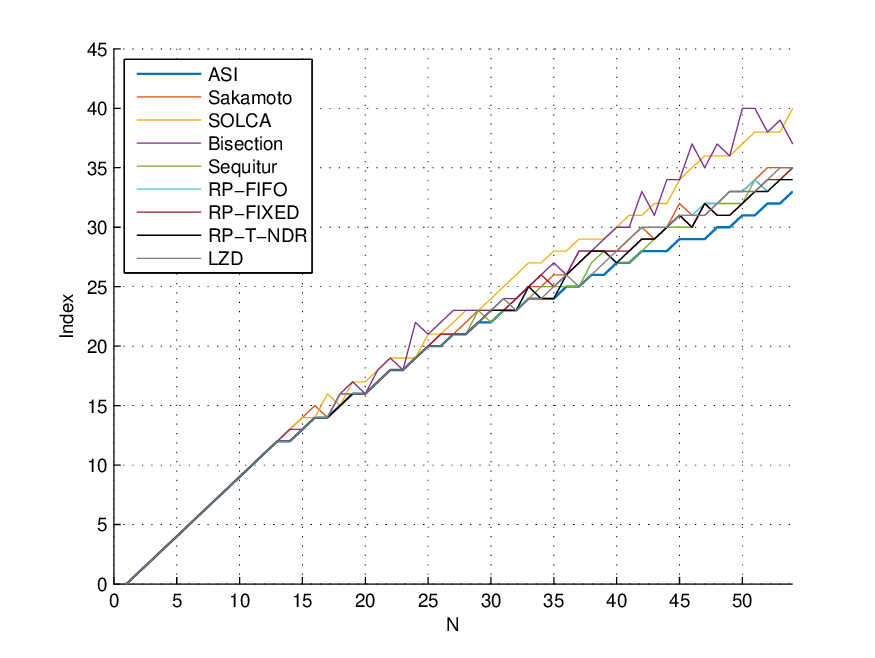}
    \caption{ASI vs.\ algorithm-specific index for the ternary maxASI strings.}
    \label{fig:Fig_b3}
\end{figure}

\begin{table}[htbp]
\centering
\small
\caption{Rank correlations and mean relative error for the ternary maxASI strings.}
\label{tab:corr_ternary_maxasi}
\begin{tabular}{|l|c|c|c|}
\hline
\textbf{Algorithm $A$} &
\textbf{Spearman's $\rho_A$} &
\textbf{Kendall's $\tau_A$} &
\textbf{Mean error $\overline{\Delta_A}$ [\%]} \\
\hline
Sakamoto      & 0.9974 & 0.9809 &  5.04 \\
SOLCA         & 0.9977 & 0.9815 & 11.43 \\
Bisection     & 0.9929 & 0.9553 & 10.77 \\
Sequitur      & 0.9960 & 0.9761 &  2.24 \\
Re-Pair FIFO  & 0.9967 & 0.9785 &  2.94 \\
Re-Pair FIXED & 0.9970 & 0.9809 &  3.22 \\
Re-Pair T-NDR & 0.9962 & 0.9767 &  2.46 \\
LZD           & 0.9985 & 0.9887 &  3.14 \\
\hline
\end{tabular}
\end{table}

On the ternary set ($42$ strings, $13 \leq N \leq 54$), LZD leads in rank correlation ($\rho = 0.9985$), while Sequitur achieves the lowest error ($2.24\%$). Sakamoto recovers strongly: its error drops from $\sim$19\% (binary) to $5.04\%$, suggesting that ternary maxASI strings contain more structured repetitions aligned with Sakamoto's run-decomposition phase. SOLCA and Bisection maintain high correlations ($\rho > 0.99$) despite double-digit errors---a decoupling between ranking and precision.

\begin{figure}[htbp]
    \centering
    \includegraphics[width=0.9\textwidth]{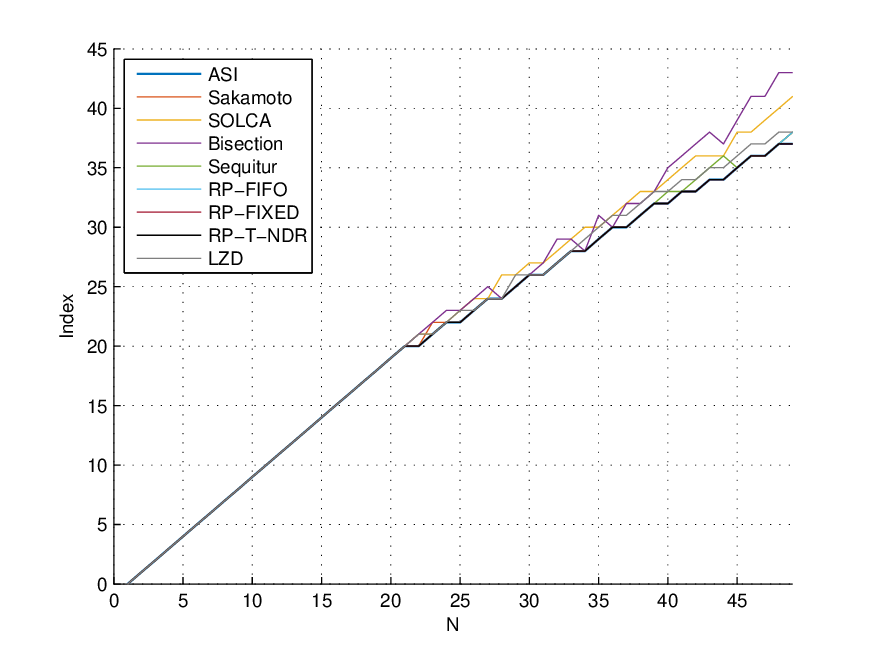}
    \caption{ASI vs.\ algorithm-specific index for the quaternary maxASI strings.}
    \label{fig:Fig_b4}
\end{figure}

\begin{table}[htbp]
\centering
\small
\caption{Rank correlations and mean relative error for the quaternary maxASI strings.}
\label{tab:corr_quaternary_maxasi}
\begin{tabular}{|l|c|c|c|}
\hline
\textbf{Algorithm $A$} &
\textbf{Spearman's $\rho_A$} &
\textbf{Kendall's $\tau_A$} &
\textbf{Mean error $\overline{\Delta_A}$ [\%]} \\
\hline
Sakamoto      & 0.9995 & 0.9962 & 0.26 \\
SOLCA         & 0.9968 & 0.9785 & 5.03 \\
Bisection     & 0.9962 & 0.9735 & 6.38 \\
Sequitur      & 0.9970 & 0.9848 & 0.61 \\
Re-Pair FIFO  & 0.9999 & 0.9987 & 0.09 \\
Re-Pair FIXED & 1.0000 & 1.0000 & 0.00 \\
Re-Pair T-NDR & 1.0000 & 1.0000 & 0.00 \\
LZD           & 0.9988 & 0.9911 & 2.15 \\
\hline
\end{tabular}
\end{table}

For the quaternary set ($29$ strings, $21 \leq N \leq 49$), Re-Pair FIXED and T-NDR achieve exact agreement with ASI ($\rho = \tau = 1.000$, $0.00\%$ error). Sakamoto and Sequitur also perform strongly (errors $0.26\%$ and $0.61\%$), reversing their binary behaviour. SOLCA and Bisection remain the least precise ($\sim$5--6\% error) despite high correlations.

\subsubsection{83 Unary Strings}

\begin{figure}[htbp]
    \centering
    \includegraphics[width=0.9\textwidth]{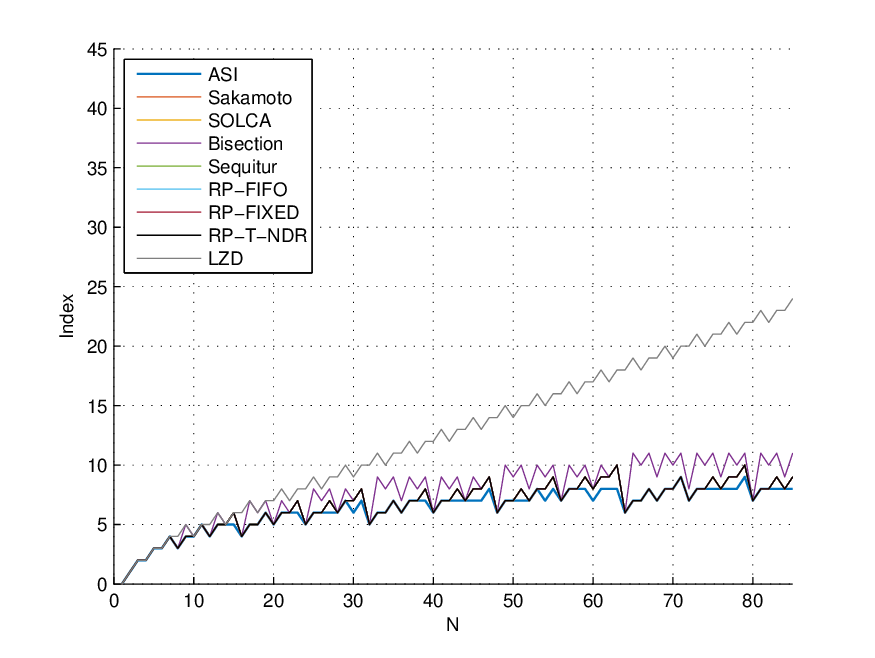}
    \caption{ASI vs.\ algorithm-specific index for unary strings.}
    \label{fig:Fig_b1}
\end{figure}

\begin{table}[htbp]
\centering
\small
\caption{Rank correlations and mean relative error for unary strings.}
\label{tab:corr_unary}
\begin{tabular}{|l|c|c|c|}
\hline
\textbf{Algorithm $A$} &
\textbf{Spearman's $\rho_A$} &
\textbf{Kendall's $\tau_A$} &
\textbf{Mean error $\overline{\Delta_A}$ [\%]} \\
\hline
Sakamoto      & 0.9515 & 0.9133 & 4.16 \\
SOLCA         & 0.9515 & 0.9133 & 4.16 \\
Bisection     & 0.9152 & 0.8465 & 19.15 \\
Sequitur      & 0.9515 & 0.9133 & 4.16 \\
Re-Pair FIFO  & 0.9515 & 0.9133 & 4.16 \\
Re-Pair FIXED & 0.9515 & 0.9133 & 4.16 \\
Re-Pair T-NDR & 0.9515 & 0.9133 & 4.16 \\
LZD           & 0.9004 & 0.7996 & 91.39 \\
\hline
\end{tabular}
\end{table}

Six algorithms (Sakamoto, SOLCA, Sequitur, and all Re-Pair variants) yield identical results on unary strings ($\rho = 0.9515$, error $4.16\%$). LZD is an outlier: $\rho = 0.9004$ but $91.39\%$ error, indicating that while it captures the monotonic trend, it fails to estimate absolute grammar size.

Hucke et al.~\cite{hucke_approximation_2021} showed that Re-Pair on unary strings $a^N$ produces a grammar of size $\lfloor \log_2 N \rfloor + \nu(N) - 1$, where $\nu(N)$ is the Hamming weight of $N$. This equals the binary addition chain method (Chandah-sutra~\cite{A014701}). Our experiments confirm that Sakamoto, SOLCA, Sequitur, and all Re-Pair variants produce identical grammar sizes on the same strings. The $4.16\%$ mean error therefore reflects the gap between the binary addition chain and the optimal addition chain~\cite{A003313}---a limitation independent of heuristic strategy. The binary addition chain first deviates from the optimal at $N=15$ (binary chain length 6 vs.\ optimal 5: $1,2,3,6,12,15$). The largest gap in the tested range occurs at $N=63$ (binary 10 vs.\ optimal 8); overall, the two methods disagree for 24 of the 83 values.

\subsubsection{Alphabet-Size Effect and Divergence Point}

\begin{figure}[htbp]
    \centering
    \includegraphics[width=0.9\textwidth]{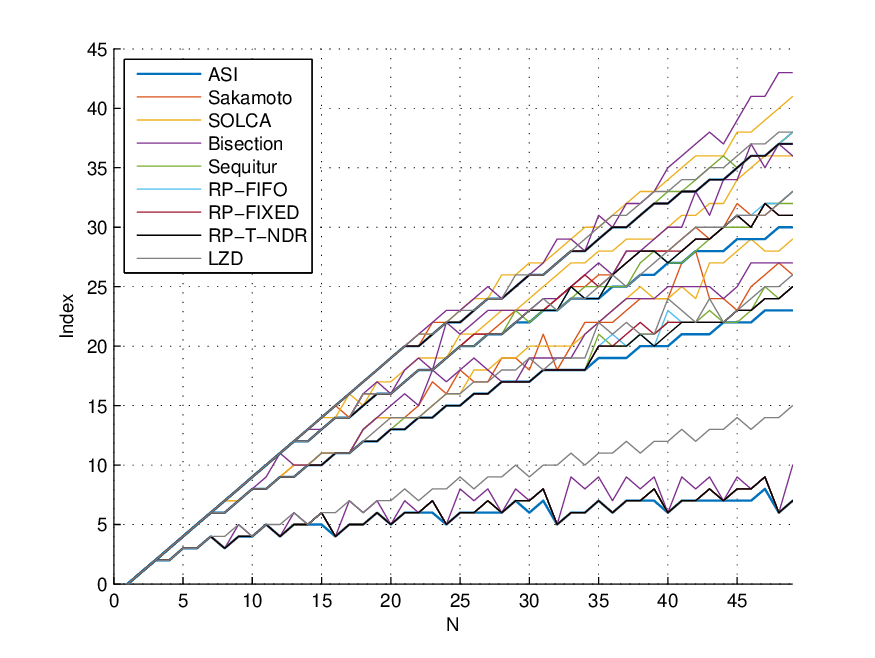}
    \caption{Compression algorithms vs.\ ASI ($|\Sigma| \in \{1,2,3,4\}$).}
    \label{fig:Fig_ball}
\end{figure}

While all algorithms follow monotonic growth relative to $N$, absolute values diverge with increasing string length. The relative performance ranking remains stable across alphabet sizes, but approximation quality improves as $|\Sigma|$ increases.

\begin{definition}[Divergence Point]\label{def:divergence_point}
For an algorithm $A$ and alphabet size $|\Sigma|$, the \emph{divergence point} $N^*(A, |\Sigma|)$ is the smallest $N$ for which there exists a string $w \in \Sigma^N$ with $L_A(w) > \mathrm{ASI}(w)$. Below this threshold, $L_A(w) = \mathrm{ASI}(w)$ for all $w$.
\end{definition}

In binary regimes, constrained combinatorial space causes frequent digram collisions and suboptimal branching early in assembly, leading to divergence at small $N$. Larger alphabets provide greater ``combinatorial resolution'', delaying divergence. This shift is visible in the maxASI figures across all alphabet sizes (except unary, where all algorithms except LZD and Bisection are equivalent).

\begin{table}[h!]
\centering
\caption{Re-Pair T-NDR performance across maxASI strings.}
\label{tab:alphabet_effect}
\small
\begin{tabular}{@{}lcc@{}}
\toprule
\textbf{Alphabet Size} & \textbf{Mean Error} & \textbf{Exact Match Rate} \\ \midrule
Binary ($|\Sigma|=2$)     & $4.03\%$ & $40.51\%$ \\ 
Ternary ($|\Sigma|=3$)    & $2.46\%$ & $52.38\%$ \\ 
Quaternary ($|\Sigma|=4$) & $0.00\%$ & $100.00\%$ \\ 
\bottomrule
\end{tabular}
\end{table}

The same trend holds for Sequitur (exact match rates: $36.71\%$ binary, $59.52\%$ ternary, $82.76\%$ quaternary) and Sakamoto ($11.39\%$, $26.19\%$, $93.10\%$). Sakamoto exhibits the strongest alphabet dependence: nearly ineffective on binary strings, it becomes competitive on quaternary.

\subsection{Random Strings}

The random set---spanning low- to high-entropy strings, both random-length and fixed-length ($N=30$) with stratified entropy intervals---reveals a different performance regime. All algorithms exhibit higher mean errors than on maxASI strings, and rank correlations vary across the entropy spectrum.

Data were pre-processed by grouping points by Shannon entropy $H$ and computing mean absolute error per group, sorted by ascending $H$. We acknowledge that $N=30$ is short for unequivocal generalisations.

\subsubsection{Random Binary Strings (20 strings, varied length)}

\begin{figure}[htbp]
    \centering
    \includegraphics[width=0.9\linewidth]{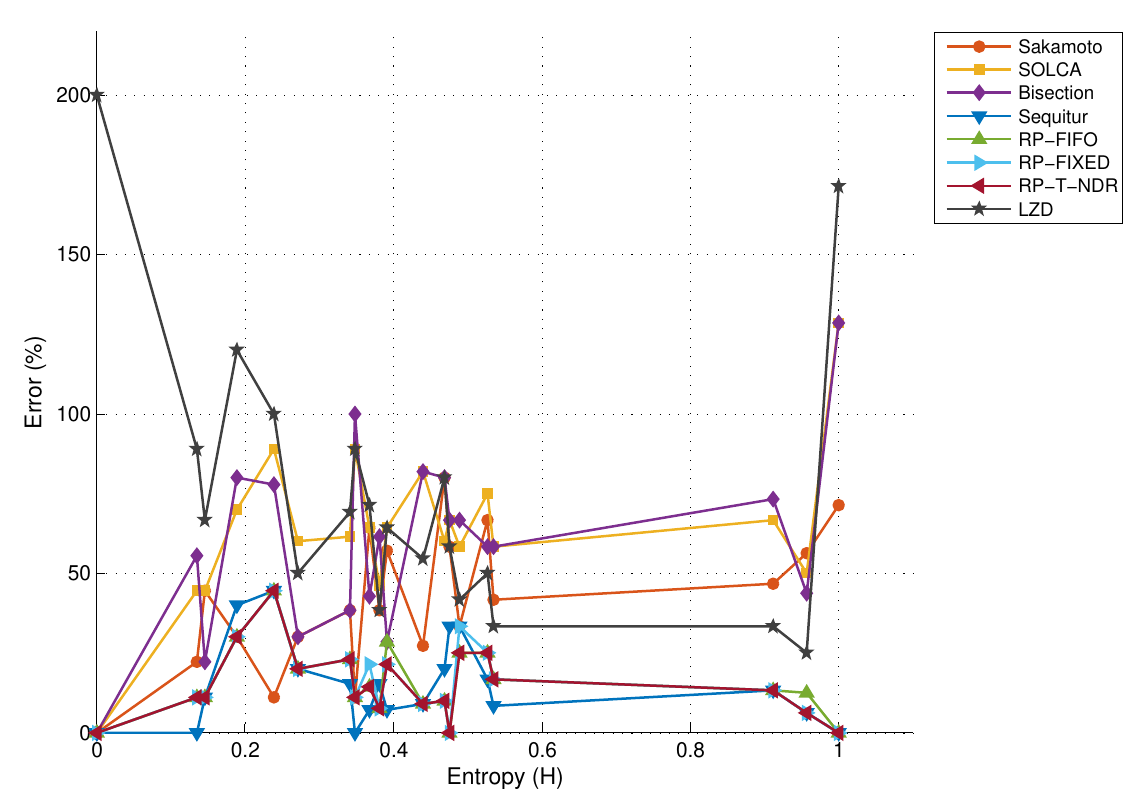}
    \caption{ASI estimation error vs.\ Shannon entropy for random binary strings.}
    \label{fig:entropy_error_binary}
\end{figure}

\begin{table}[htbp]
\centering
\small
\caption{Rank correlation and mean relative error for random binary strings. Significance stars report uncorrected $p$-values for individual tests; 
see Section~\ref{sec:metrics} for multiple-comparison policy.}
\label{tab:corr_binary_entropy}
\begin{tabular}{|l|c|c|c|c|}
\hline
\textbf{Algorithm} &
\textbf{$\rho_A$} &
\textbf{$\tau_A$} &
\textbf{$\overline{\Delta_A}$ [\%]} &
\textbf{$\rho_{(H,\Delta_A)}$} \\
\hline
Sakamoto      & 0.9272*** & 0.8228*** & 41.44 & 0.6564** \\
SOLCA         & 0.9171*** & 0.8058*** & 63.92 & 0.3097 \\
Bisection     & 0.8140*** & 0.6725*** & 59.72 & 0.3892* \\
Sequitur      & 0.8998*** & 0.7909*** & 15.05 & $-$0.0242 \\
Re-Pair FIFO  & 0.9271*** & 0.8410*** & 15.65 & $-$0.0928 \\
Re-Pair FIXED & 0.9265*** & 0.8391*** & 15.75 & $-$0.1780 \\
Re-Pair T-NDR & 0.9286*** & 0.8458*** & 14.98 & $-$0.1848 \\
LZD           & 0.5131*   & 0.3732*   & 75.27 & $-$0.5728** \\
\hline
\end{tabular}
\\[2pt]
\small \textsuperscript{***} $p<0.001$, \textsuperscript{**} $p<0.01$, \textsuperscript{*} $p<0.05$.
\end{table}

Re-Pair T-NDR and Sequitur are the most resilient (errors $\sim$15\%). The set reveals structural mismatches: for the zero-entropy string $0^{64}$ ($\mathrm{ASI}=6$), LZD yields $L=18$ ($+200\%$), while Bisection and all Re-Pair variants match the optimum. Conversely, Bisection fails on $(01)^{34}$ ($N=68, \mathrm{ASI}=7$), producing $L=16$ ($+128\%$); only Sequitur and Re-Pair family maintain optimality.

The error-entropy Spearman coefficient divides algorithms into two groups. Sakamoto ($\rho_{(H,\Delta_A)} = 0.6564^{**}$), SOLCA ($0.3097$), and Bisection ($0.3892^{*}$) show positive correlation---error grows with entropy. Re-Pair variants and Sequitur maintain negative coefficients, with T-NDR at $-0.1848$. LZD shows the most negative correlation ($-0.5728^{**}$) but its high mean error ($75.27\%$) and poor $\rho_A$ ($0.5131^{*}$) indicate consistent inaccuracy regardless of entropy.

Sequitur and Re-Pair family maintain the lowest error rates across the entropy spectrum. LZD demonstrates extreme instability at boundaries ($200\%$ at $H=0$, $>170\%$ at $H=1$). 
SOLCA and Bisection exhibit volatile, correlated error patterns with large fluctuations at mid-range entropy.

\subsubsection{Random Fixed-Length Binary Strings (25 strings, $N=30$)}

\begin{figure}[htbp]
    \centering
    \includegraphics[width=0.9\linewidth]{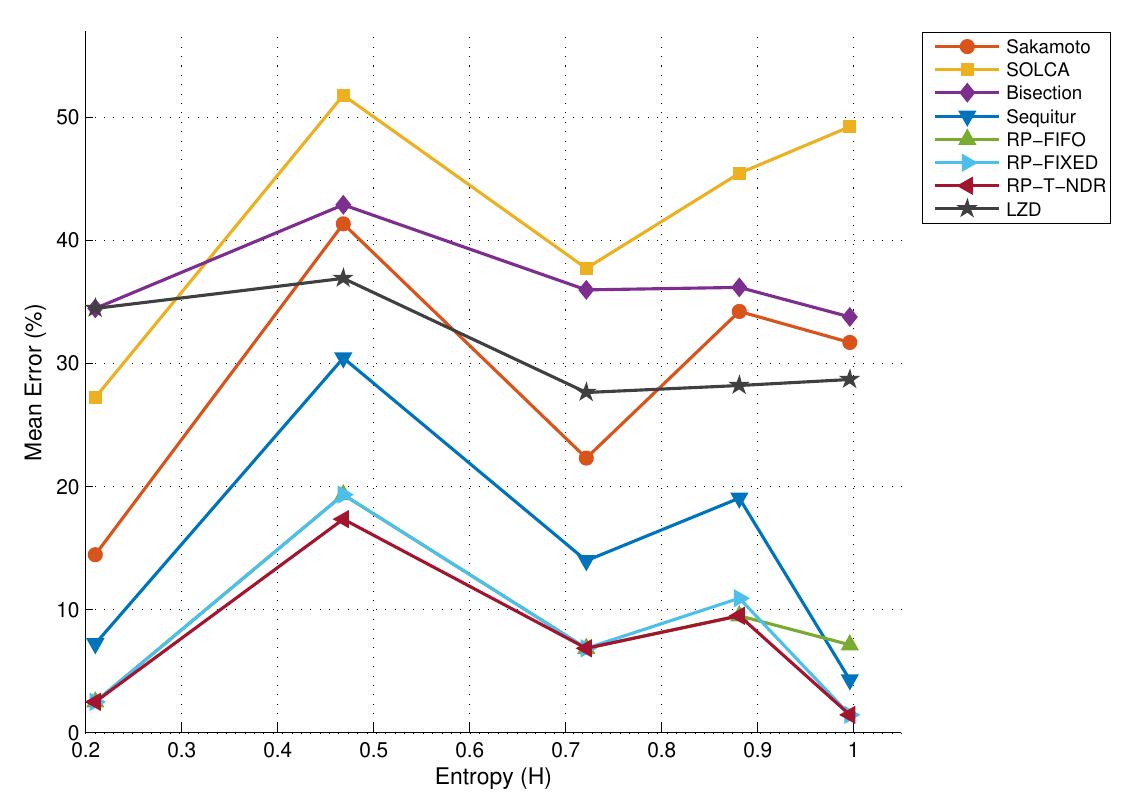}
    \caption{ASI estimation error vs.\ Shannon entropy for fixed-length binary strings ($N=30$).}
    \label{fig:entropy_error_binary_n30}
\end{figure}

\begin{table}[htbp]
\centering
\small
\caption{Rank correlation and mean relative error for fixed-length binary strings ($N=30$).}
\label{tab:corr_binary_entropyn30}
\begin{tabular}{|l|c|c|c|c|}
\hline
\textbf{Algorithm} &
\textbf{$\rho_A$} &
\textbf{$\tau_A$} &
\textbf{$\overline{\Delta_A}$ [\%]} &
\textbf{$\rho_{(H,\Delta_A)}$} \\
\hline
Sakamoto      & 0.8028*** & 0.6621*** & 28.80 & 0.2612 \\
SOLCA         & 0.8464*** & 0.7177*** & 42.27 & 0.2555 \\
Bisection     & 0.8989*** & 0.8046*** & 36.65 & $-$0.2007 \\
Sequitur      & 0.8411*** & 0.7078*** & 14.99 & $-$0.2010 \\
Re-Pair FIFO  & 0.9232*** & 0.8226*** &  9.06 & 0.0446 \\
Re-Pair FIXED & 0.9113*** & 0.8053*** &  8.21 & $-$0.1355 \\
Re-Pair T-NDR & 0.9161*** & 0.8175*** &  7.52 & $-$0.1499 \\
LZD           & 0.8997*** & 0.7887*** & 31.17 & $-$0.3405* \\
\hline
\end{tabular}
\\[2pt]
\small \textsuperscript{***} $p<0.001$, \textsuperscript{*} $p<0.05$.
\end{table}

All algorithms maintain $\rho_A > 0.8$. T-NDR achieves the lowest error ($7.52\%$). SOLCA and Bisection overestimate by large margins ($42.27\%$ and $36.65\%$). LZD improves relative to the varied-length set ($31.17\%$ vs.\ $75.27\%$) but remains less accurate than Sequitur ($14.99\%$).

Sakamoto and SOLCA retain positive entropy-error correlations ($0.2612$ and $0.2555$), confirming sensitivity to randomness even at fixed length. T-NDR maintains a slightly negative correlation ($-0.1499$). LZD shows the strongest negative correlation ($-0.3405^{*}$), supporting the observation that while consistently inaccurate, its performance is flat across entropy levels.

At low entropy ($H \approx 0.21$), Re-Pair family and Sequitur achieve near-zero error. A global peak occurs at mid-range ($H \approx 0.47$), with SOLCA exceeding $50\%$. Beyond this, Re-Pair family and Sequitur stabilise below $20\%$, while Sakamoto, LZD, and Bisection remain elevated. At maximum entropy ($H \approx 1$), Sequitur and T-NDR return to near-ideal accuracy; SOLCA remains the most discrepant.

\subsubsection{Random Ternary Strings (20 strings, varied length)}

\begin{figure}[htbp]
    \centering
    \includegraphics[width=0.9\linewidth]{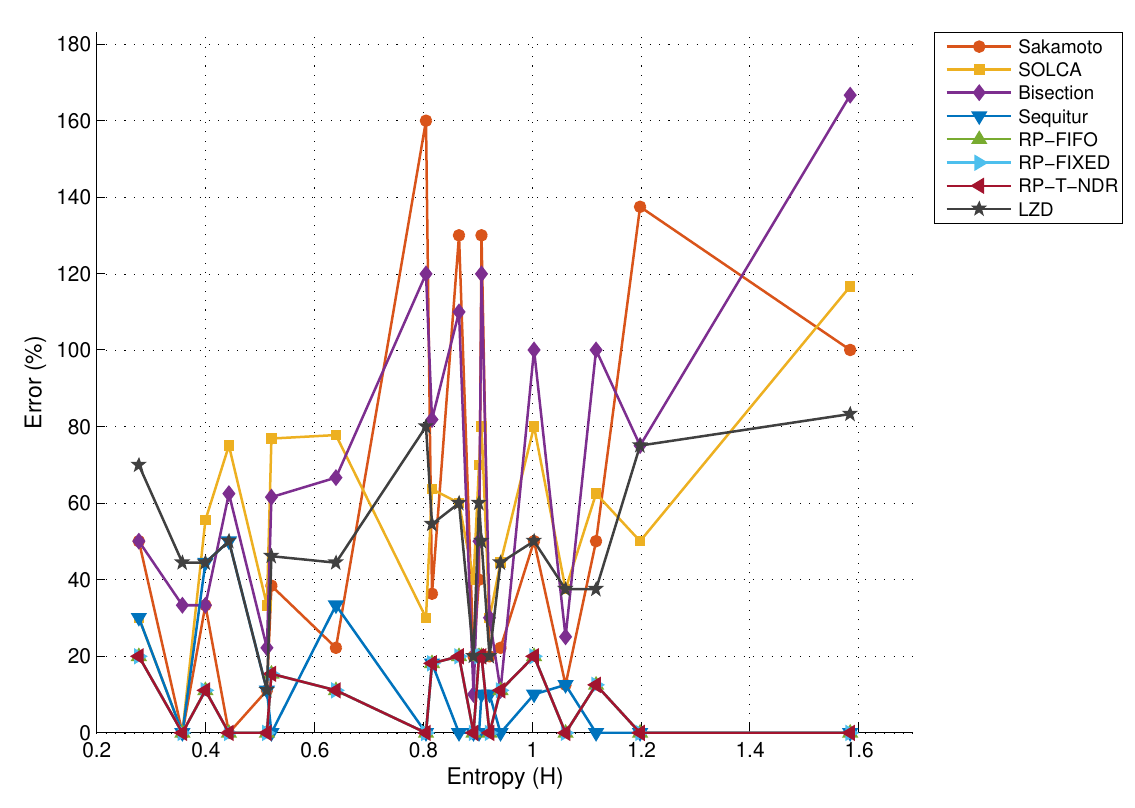}
    \caption{ASI estimation error vs.\ Shannon entropy for random ternary strings.}
    \label{fig:entropy_error_ternary}
\end{figure}

\begin{table}[htbp]
\centering
\small
\caption{Rank correlation and mean relative error for random ternary strings.}
\label{tab:corr_ternary_entropy}
\begin{tabular}{|l|c|c|c|c|}
\hline
\textbf{Algorithm} &
\textbf{$\rho_A$} &
\textbf{$\tau_A$} &
\textbf{$\overline{\Delta_A}$ [\%]} &
\textbf{$\rho_{(H,\Delta_A)}$} \\
\hline
Sakamoto      & 0.5923*  & 0.4894*  & 53.19 & 0.3771 \\
SOLCA         & 0.6472** & 0.5444** & 55.67 & 0.3217 \\
Bisection     & 0.4767   & 0.3464   & 66.46 & 0.2688 \\
Sequitur      & 0.6172*  & 0.5460*  & 11.48 & $-$0.4418* \\
Re-Pair FIFO  & 0.9253***& 0.8861***&  8.97 & $-$0.0653 \\
Re-Pair FIXED & 0.9253***& 0.8861***&  8.97 & $-$0.0653 \\
Re-Pair T-NDR & 0.9253***& 0.8861***&  8.97 & $-$0.0653 \\
LZD           & 0.7462** & 0.6319** & 49.15 & 0.0855 \\
\hline
\end{tabular}
\\[2pt]
\small \textsuperscript{***} $p<0.001$, \textsuperscript{**} $p<0.01$, \textsuperscript{*} $p<0.05$.
\end{table}

The Re-Pair family dominates ($8.97\%$ error, $\rho_A = 0.9253^{***}$). Sequitur exhibits decoupled accuracy and ranking: low error ($11.48\%$) but collapsed rank correlation ($\rho_A = 0.6172^{*}$). Sakamoto ($53.19\%$) and Bisection ($66.46\%$) perform poorly.

The ternary minASI string $(012)^9$ ($N=27, \mathrm{ASI}=6$) confirms this: all Re-Pair variants and Sequitur match the optimum, while LZD fails ($49.15\%$ error) despite moderate ranking ($\rho_A = 0.7462^{**}$).

Positive entropy-error correlations for Sakamoto ($0.3771$), SOLCA ($0.3217$), and Bisection ($0.2688$) confirm that their errors are driven by symbol distribution balance. Sequitur displays significant negative correlation ($-0.4418^{*}$). Re-Pair family maintains near-neutral balance ($-0.0653$).

As uniformity increases, Sakamoto, SOLCA, and Bisection exhibit substantial error spikes (Bisection reaches $166.67\%$ at highest entropy). Sequitur and Re-Pair family maintain near-zero error for several high-entropy points. A pronounced divergence occurs at $H \approx 0.8$--$1.2$, where Sakamoto and Bisection errors fluctuate aggressively.

\subsubsection{Random Fixed-Length Ternary Strings (25 strings, $N=30$)}

\begin{figure}[htbp]
    \centering
    \includegraphics[width=0.9\linewidth]{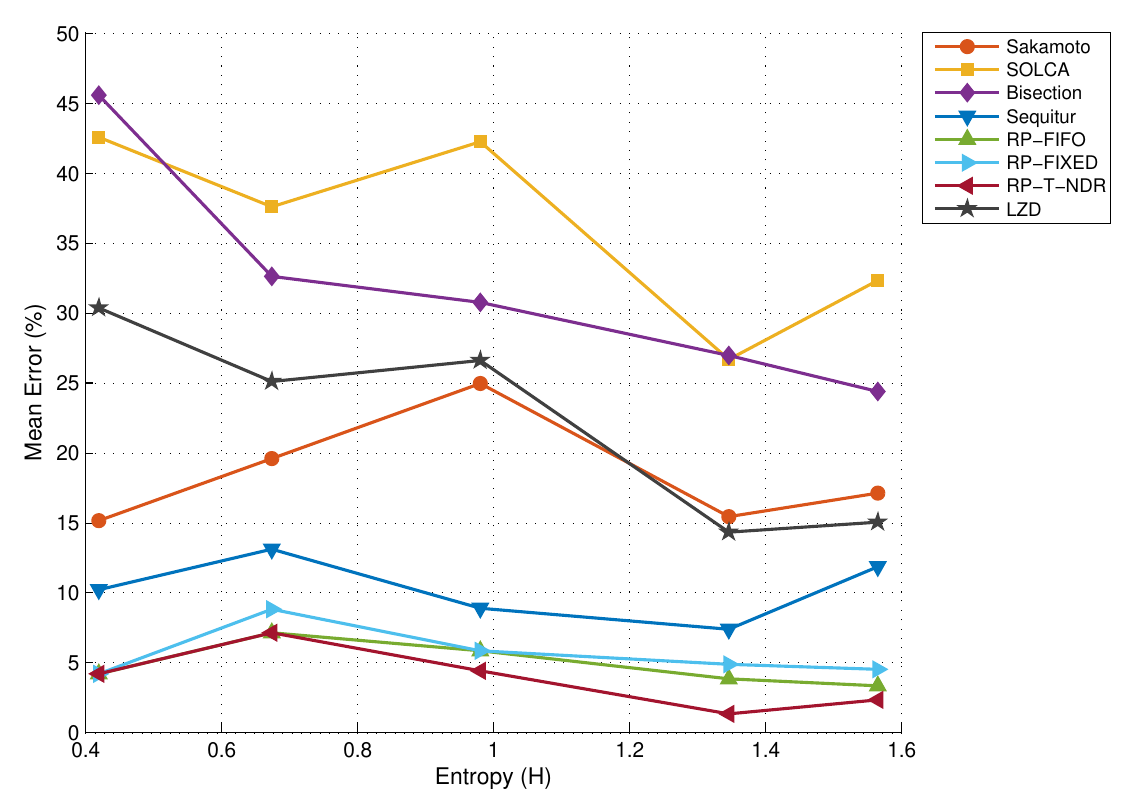}
    \caption{ASI estimation error vs.\ Shannon entropy for fixed-length ternary strings ($N=30$).}
    \label{fig:entropy_error_ternary_n30}
\end{figure}

\begin{table}[htbp]
\centering
\small
\caption{Rank correlation and mean relative error for fixed-length ternary strings ($N=30$).}
\label{tab:corr_ternary_entropyn30}
\begin{tabular}{|l|c|c|c|c|}
\hline
\textbf{Algorithm} &
\textbf{$\rho_A$} &
\textbf{$\tau_A$} &
\textbf{$\overline{\Delta_A}$ [\%]} &
\textbf{$\rho_{(H,\Delta_A)}$} \\
\hline
Sakamoto      & 0.9573*** & 0.8607*** & 18.47 & $-$0.0353 \\
SOLCA         & 0.9493*** & 0.8663*** & 36.31 & $-$0.3717* \\
Bisection     & 0.8798*** & 0.7362*** & 32.09 & $-$0.4556* \\
Sequitur      & 0.9383*** & 0.8375*** & 10.30 & 0.0436 \\
Re-Pair FIFO  & 0.9846*** & 0.9465*** &  4.88 & $-$0.1805 \\
Re-Pair FIXED & 0.9825*** & 0.9381*** &  5.66 & $-$0.1366 \\
Re-Pair T-NDR & 0.9873*** & 0.9534*** &  3.89 & $-$0.2896 \\
LZD           & 0.9468*** & 0.8612*** & 22.31 & $-$0.4888** \\
\hline
\end{tabular}
\\[2pt]
\small \textsuperscript{***} $p<0.001$, \textsuperscript{**} $p<0.01$, \textsuperscript{*} $p<0.05$.
\end{table}

Most algorithms achieve $\rho_A > 0.93^{***}$. T-NDR remains the most accurate ($3.89\%$). SOLCA and Bisection overestimate ($36.31\%$ and $32.09\%$). Sakamoto improves relative to the varied-length set ($18.47\%$ vs.\ $53.19\%$).

Entropy-error correlations are predominantly negative: LZD ($-0.4888^{**}$), Bisection ($-0.4556^{*}$), SOLCA ($-0.3717^{*}$). T-NDR ($-0.2896$) combines low error with negative correlation. Sakamoto ($-0.0353$) and Sequitur ($0.0436$) are near-neutral.

SOLCA and Bisection consistently represent the upper error bound ($25$--$45\%$). Sequitur and Re-Pair family maintain precision, often below $10\%$ as entropy rises. T-NDR approaches near-zero error at $H \approx 1.35$.

\subsubsection{Random Quaternary Strings (20 strings, varied length)}

\begin{figure}[htbp]
    \centering
    \includegraphics[width=0.9\linewidth]{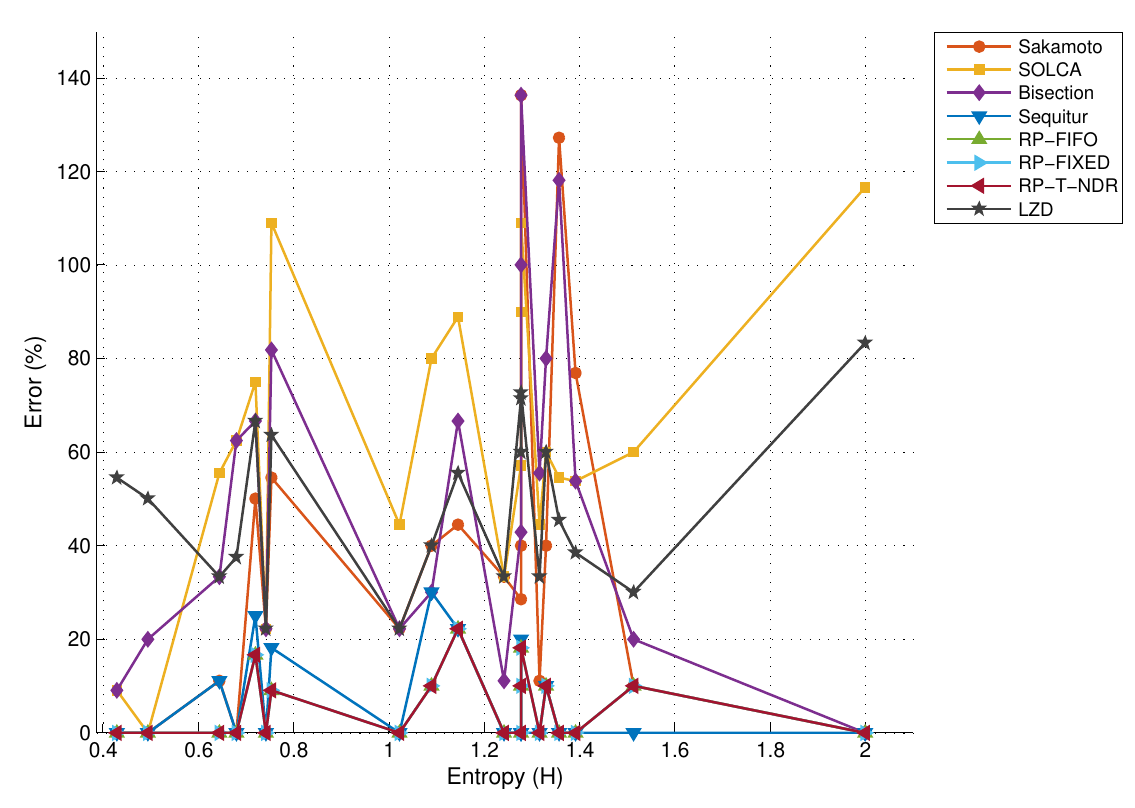}
    \caption{ASI estimation error vs.\ Shannon entropy for random quaternary strings.}
    \label{fig:entropy_error_quaternary}
\end{figure}

\begin{table}[htbp]
\centering
\small
\caption{Rank correlation and mean relative error for random quaternary strings.}
\label{tab:corr_quaternary_entropy}
\begin{tabular}{|l|c|c|c|c|}
\hline
\textbf{Algorithm} &
\textbf{$\rho_A$} &
\textbf{$\tau_A$} &
\textbf{$\overline{\Delta_A}$ [\%]} &
\textbf{$\rho_{(H,\Delta_A)}$} \\
\hline
Sakamoto      & 0.8224*** & 0.7133*** & 37.41 & 0.2947 \\
SOLCA         & 0.6331*   & 0.5038*   & 61.29 & 0.3272 \\
Bisection     & 0.7365**  & 0.6115**  & 51.62 & 0.1568 \\
Sequitur      & 0.8690*** & 0.7893*** &  6.83 & $-$0.2103 \\
Re-Pair FIFO  & 0.9297*** & 0.8833*** &  5.31 & 0.1255 \\
Re-Pair FIXED & 0.9297*** & 0.8833*** &  5.31 & 0.1255 \\
Re-Pair T-NDR & 0.9297*** & 0.8833*** &  5.31 & 0.1255 \\
LZD           & 0.9137*** & 0.8140*** & 48.69 & 0.1562 \\
\hline
\end{tabular}
\\[2pt]
\small \textsuperscript{***} $p<0.001$, \textsuperscript{**} $p<0.01$, \textsuperscript{*} $p<0.05$.
\end{table}

Re-Pair family defines the performance ceiling ($5.31\%$ error, $\rho_A = 0.9297^{***}$). LZD exhibits decoupled ranking and precision: high $\rho_A$ ($0.9137^{***}$) but $48.69\%$ error. Sequitur achieves low error ($6.83\%$) but trails in rank correlation ($0.8690^{***}$).

SOLCA shows moderate positive entropy-error correlation ($0.3272$), peaking at $116.67\%$ at $H=2.0$. LZD ($0.1562$) and Sakamoto ($0.2947$) display weaker positive correlations. Sequitur is negative ($-0.2103$). The Re-Pair family is slightly positive ($0.1255$) but with the lowest mean error, indicating robustness against alphabet-driven noise.

Error peaks concentrate for $H \in [1.2, 1.4]$, where Sakamoto and Bisection exceed $130\%$. At $H=2.0$, SOLCA reaches $116.67\%$, while Sequitur and the Re-Pair family converge to $0\%$ error.

\subsubsection{Random Fixed-Length Quaternary Strings (25 strings, $N=30$)}

\begin{figure}[htbp]
    \centering
    \includegraphics[width=0.9\linewidth]{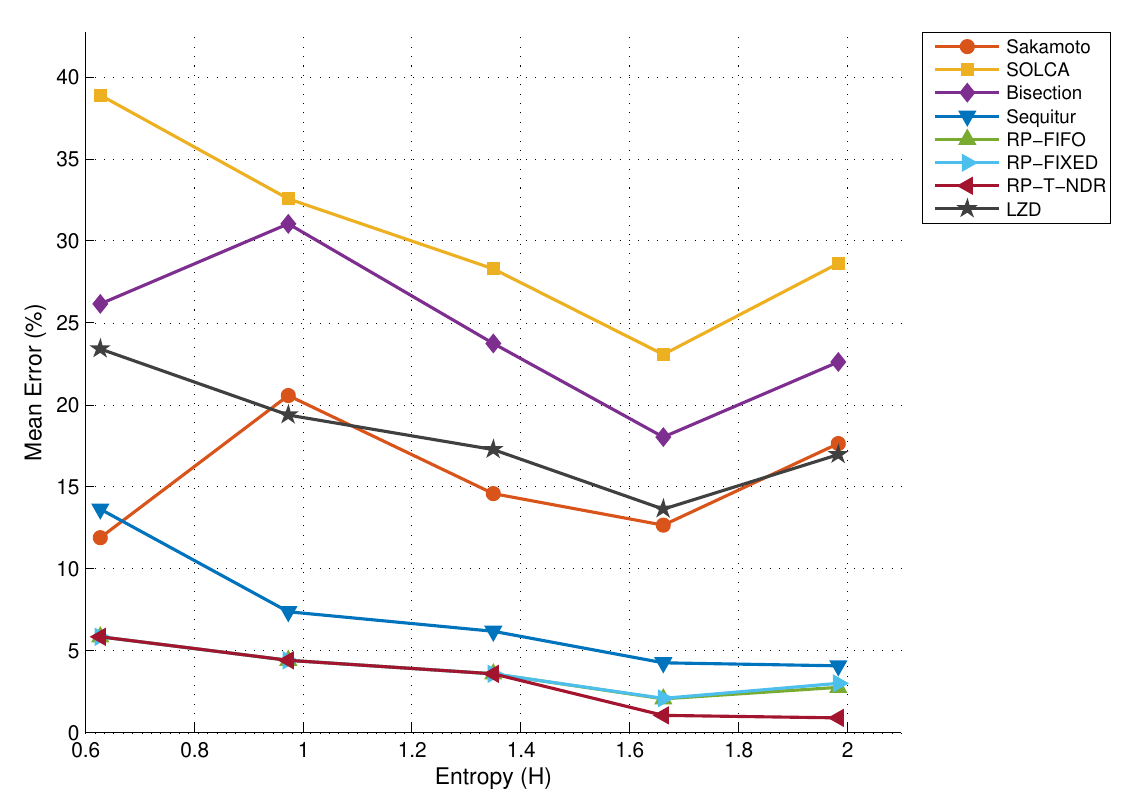}
    \caption{ASI estimation error vs.\ Shannon entropy for fixed-length quaternary strings ($N=30$).}
    \label{fig:entropy_error_quaternary_n30}
\end{figure}

\begin{table}[htbp]
\centering
\small
\caption{Rank correlation and mean relative error for fixed-length quaternary strings ($N=30$).}
\label{tab:corr_quaternary_entropyn30}
\begin{tabular}{|l|c|c|c|c|}
\hline
\textbf{Algorithm} &
\textbf{$\rho_A$} &
\textbf{$\tau_A$} &
\textbf{$\overline{\Delta_A}$ [\%]} &
\textbf{$\rho_{(H,\Delta_A)}$} \\
\hline
Sakamoto      & 0.9657*** & 0.8901*** & 15.47 & 0.1336 \\
SOLCA         & 0.9545*** & 0.8612*** & 30.30 & $-$0.3515* \\
Bisection     & 0.9402*** & 0.8377*** & 24.32 & $-$0.3178* \\
Sequitur      & 0.9667*** & 0.8979*** &  7.11 & $-$0.4078* \\
Re-Pair FIFO  & 0.9874*** & 0.9507*** &  3.74 & $-$0.2922* \\
Re-Pair FIXED & 0.9816*** & 0.9331*** &  3.80 & $-$0.2751* \\
Re-Pair T-NDR & 0.9872*** & 0.9524*** &  3.17 & $-$0.4475** \\
LZD           & 0.9604*** & 0.8834*** & 18.14 & $-$0.3494* \\
\hline
\end{tabular}
\\[2pt]
\small \textsuperscript{***} $p<0.001$, \textsuperscript{**} $p<0.01$, \textsuperscript{*} $p<0.05$.
\end{table}

Nearly all heuristics achieve $\rho_A > 0.94^{***}$. Re-Pair FIFO and T-NDR exceed $0.987^{***}$. T-NDR remains the most accurate ($3.17\%$). Sakamoto ($15.47\%$) and LZD ($18.14\%$) improve; SOLCA ($30.30\%$) and Bisection ($24.32\%$) remain the most discrepant.

Entropy-error correlations are strongly negative across almost all heuristics. T-NDR shows the most pronounced negative correlation ($-0.4475^{**}$), coupled with the lowest error. Sequitur ($-0.4078^{*}$), SOLCA ($-0.3515^{*}$), Bisection ($-0.3178^{*}$), and LZD ($-0.3494^{*}$) follow. Sakamoto is the only near-neutral heuristic ($0.1336$), indicating constant error regardless of symbol distribution balance.

Mean error declines with increasing entropy, particularly for $H \in [1.0, 1.7]$. Sequitur and Re-Pair family show the greatest stability; T-NDR falls below $2\%$ at $H \approx 1.66$ and $H \approx 1.98$.

\subsection{Biological Data}

To evaluate GBCA performance in a biological context, we compiled 40 strings: 20 polypeptide chains from RCSB PDB and 20 nucleic acid sequences from NCBI GenBank. Proteins use $|\Sigma| \approx 20$; nucleic acids use $|\Sigma| = 4$.

Table~\ref{tab:biological_results} presents a representative subset; full results are in Table~\ref{tab:biological_results_all} in Appendix.

\begin{table}[htbp]
\centering
\caption{Compression results for representative biological sequences. `n.a.' indicates ASI computation exceeded limits.}
\label{tab:biological_results}
\small
\begin{tabular}{|l|c|c|c|c|c|c|c|c|c|c|c|}
\hline
\textbf{ID} & \textbf{$N$} & \textbf{$|\Sigma|$} & \textbf{ASI} & \textbf{Sak.} & \textbf{SOL.} & \textbf{Bis.} & \textbf{Seq.} & \textbf{FIFO} & \textbf{FIX.} & \textbf{T-NDR} & \textbf{LZD} \\
\hline
\multicolumn{12}{|l|}{\textit{Polypeptide Chains}} \\
\hline
1UBQ & 76 & 18 & 64 & 64 & 71 & 75 & 67 & 64 & 64 & 64 & 68 \\
1BNR & 110 & 18 & n.a. & 97 & 102 & 105 & 96 & 96 & 96 & 96 & 99 \\
3GBN & 179 & 20 & n.a. & 158 & 170 & 174 & 156 & 154 & 154 & 153 & 164 \\
\hline
\multicolumn{12}{|l|}{\textit{Nucleic Acid Sequences}} \\
\hline
NR\_029492 & 71 & 4 & 36 & 43 & 49 & 49 & 42 & 40 & 40 & 37 & 44 \\
NR\_033048 & 110 & 4 & n.a. & 66 & 74 & 76 & 65 & 65 & 64 & 63 & 68 \\
NR\_002716 & 187 & 4 & n.a. & 105 & 116 & 120 & 102 & 103 & 105 & 101 & 111 \\
\hline
\end{tabular}
\end{table}

Re-Pair T-NDR consistently provides the tightest upper bound on ASI. For benchmark sequences with known ASI (1UBQ, NR\_029492), T-NDR reaches the global optimum or stays within $2.7\%$. Against deterministic baselines, gains are modest due to a shared algorithmic foundation. For 1UBQ, T-NDR matches FIFO/FIXED ($L=64$); across the full set, T-NDR improves in 16 of 40 cases (8 proteins, 8 nucleic acids), with reductions of 1--3 steps (e.g., 3-step reduction for NR\_029492 and NR\_029856 vs.\ best deterministic output).

This set serves as proof-of-concept; biological conclusions require larger samples and structural validation.

\subsection{Partial Spearman Analysis: Disentangling Length and Complexity}

The empirical results confirm partial decoupling between Spearman's $\rho$ and absolute numerical error. SOLCA exhibits persistent positive bias: high $\rho$ despite errors frequently exceeding $40\%$. LZD is unstable at entropy boundaries ($200\%$ at $H=0$, $171.43\%$ at $H=1$) but effective across broader spectra ($\rho > 0.90$). Bisection and Sakamoto fail in balanced-alphabet regimes, with errors above $130\%$ and degraded $\rho$.
Re-Pair family and Sequitur achieve the most favourable balance: low numerical errors (often $<10\%$) and high rank preservation.

To assess whether high $\rho$ values on maxASI sets are inflated by shared monotonic growth of ASI and $L_A(w)$ with $N$, we computed partial Spearman correlations controlling for $N$ via regression residuals. This inflation is most pronounced in maxASI sets, where $N$ increases monotonically.

\begin{table}[h!]
\centering
\caption{Standard vs.\ partial Spearman $\rho$ (controlling for $N$) on maxASI sets.}
\label{tab:partial_spearman}
\small
\setlength{\tabcolsep}{12pt}
\begin{tabular}{|l|l|c|c|}
\hline
\textbf{Set} & \textbf{Algorithm} & \textbf{Std $\rho$} & \textbf{Partial $\rho$} \\
\hline
\multirow{4}{*}{Binary maxASI} 
 & Sakamoto   & 0.992 & 0.085\textsuperscript{ns}  \\
 & Bisection  & 0.988 & $-$0.075\textsuperscript{ns} \\
 & RP-T-NDR   & 0.997 & 0.409*** \\
 & LZD        & 0.996 & $-$0.046\textsuperscript{ns} \\
\hline
\multirow{4}{*}{Ternary maxASI}
 & Sakamoto   & 0.997 & 0.297\textsuperscript{ns} \\
 & RP-FIFO    & 0.997 & 0.173\textsuperscript{ns} \\
 & RP-T-NDR   & 0.996 & 0.208\textsuperscript{ns} \\
 & LZD        & 0.998 & 0.514*** \\
\hline
\multirow{4}{*}{Quaternary maxASI}
 & Sakamoto   & 1.000 & 0.836*** \\
 & RP-FIXED   & 1.000 & 1.000*** \\
 & RP-T-NDR   & 1.000 & 1.000*** \\
 & LZD        & 0.999 & 0.584* \\
\hline
\multicolumn{4}{l}{\footnotesize *** $p < 0.001$, * $p < 0.05$, ns = not significant} \\
\end{tabular}
\end{table}

Binary maxASI ($7 \leq N \leq 85$, beyond divergence point): Standard $\rho$ is substantially inflated. Partial $\rho$ collapses to near zero for most algorithms (Sakamoto: $0.085$ ns; LZD: $-0.046$ ns), confirming that apparent rank preservation is largely an artefact of length covariation.

Ternary maxASI ($13 \leq N \leq 54$): Intermediate regime. The smaller range means a larger proportion lies below the ternary divergence point; strings exceeding it still contribute length covariation. Both mechanisms operate simultaneously, and partial $\rho$ remains non-significant for most algorithms.

Quaternary maxASI ($21 \leq N \leq 49$): Pre-divergence mechanism dominates. RP-FIXED and T-NDR achieve $\rho_{\text{partial}} = 1.000$ ($p < 0.001$), but this reflects exact reproduction ($0.00\%$ error, Table~\ref{tab:alphabet_effect}) rather than improvement in heuristic performance. Longer quaternary strings, were their ASI computable, would reveal the same divergence as binary.

These findings require reinterpretation of the high standard $\rho$ on all maxASI sets. A fully controlled comparison across alphabet sizes would require identical $N$-ranges; this was computationally infeasible.

For random sets, residual covariation between $N$ and ASI is weaker than in maxASI sets but not negligible: $\rho(\text{ASI}, N)$ rises from $\approx 0$ (binary) to moderate (ternary) to substantial (quaternary). For the Re-Pair family, this has negligible practical effect---partial $\rho$ remains highly significant. For weaker heuristics (Sakamoto, SOLCA, Bisection), a genuine length-driven component persists, reducing several partial correlations to non-significance.

The most informative evidence for length-independent complexity discrimination comes from fixed-length sets, where $N$ carries no variance and standard $\rho$ equals partial $\rho$ exactly. On these sets, T-NDR achieves $\rho \geq 0.916$ (binary), $\rho \geq 0.987$ (ternary), and $\rho \geq 0.987$ (quaternary).

\subsection{Statistical Validation: Wilcoxon Signed-Rank Tests}

To determine whether differences in mean relative error are statistically significant, we performed pairwise Wilcoxon signed-rank tests on ten set-level mean errors ($m=10$ paired observations, one per experimental set).

Per-algorithm mean errors (Table~\ref{tab:wilcoxon}) are unweighted averages of the ten set-level means, assigning equal weight to each structural regime. This is appropriate because each set represents a distinct structural regime, and the non-normal distribution of errors paired with identical benchmarking justifies the non-parametric approach.

With $m=10$, statistical power is limited; nonetheless, the magnitude of differences is sufficiently large to yield significant results.

Tie handling follows standard convention: pairs with identical set-level mean errors are excluded. On the unary set, Sakamoto, SOLCA, and Sequitur tie with T-NDR ($4.16\%$), reducing effective sample size to $m_{\text{eff}} = 9$ for those comparisons. Bisection and LZD do not tie ($19.15\%$ and $91.39\%$), so $m_{\text{eff}} = 10$. Reported $p$-values are consistent: Sakamoto and SOLCA reach $p = 0.0039 = 2/2^9$ (minimum two-sided $p$ with $m=9$, all nine non-tied differences favour T-NDR); Bisection and LZD reach $p = 0.0020 = 2/2^{10}$ (minimum with $m=10$). Sequitur also has $m_{\text{eff}} = 9$ but $p = 0.0117$ exceeds the minimum, reflecting that not all differences uniformly favour T-NDR---consistent with Sequitur's substantially lower mean error ($7.81\%$) compared to other baselines.

$p$-values were adjusted using Holm-Bonferroni step-down correction across five baseline algorithms. Rank-biserial correlation $r$ serves as effect size.

Aggregating to the set level discards within-set pairing and most statistical resolution; a per-string paired analysis with set membership as blocking factor would be more informative but requires mixed-model machinery beyond this study. The large effect sizes ($r = 1.00$ against Sakamoto, SOLCA, Bisection, and LZD; $r = 0.91$ against Sequitur) suggest qualitative conclusions are robust.

RP-FIFO and RP-FIXED are excluded from this testing framework, because T-NDR exhibits
one-sided dominance over both baselines, empirically confirmed on every tested string
(Observation~\ref{obs:tndr_empirical}). A two-sided Wilcoxon signed-rank test is therefore
methodologically inappropriate.

Excluding the 83 unary strings where all Re-Pair implementations yield identical lengths, across the remaining 325 strings T-NDR achieved strictly shorter grammars than both FIFO and FIXED in 31 instances (improvements of 1--3 steps relative to $\min(\text{FIFO}, \text{FIXED})$)---a tighter bound in $9.54\%$ of relevant cases (Observation~\ref{obs:tndr_empirical}).

\begin{table}[h!]
\centering
\caption{Statistical significance and effect sizes vs.\ Re-Pair T-NDR across all tested sets (excluding RP-FIFO and RP-FIXED).}
\label{tab:wilcoxon}
\small
\begin{tabular}{|l|c|c|c|c|}
\hline
\textbf{Algorithm} & \textbf{Mean Error} & \textbf{$p$-value} 
  & \textbf{$p_{\text{Holm}}$} & \textbf{Effect size $r$} \\
\hline
Sakamoto          & 22.28\% & 0.0039 & 0.0117 & 1.00 (Large) \\
SOLCA             & 33.08\% & 0.0039 & 0.0117 & 1.00 (Large) \\
Bisection         & 32.44\% & 0.0020 & 0.0100 & 1.00 (Large) \\
Sequitur          &  7.81\% & 0.0117 & 0.0117 & 0.91 (Large) \\
\textbf{RP-T-NDR}  & \textbf{5.45\%} & --- & --- & --- \\
LZD               & 35.22\% & 0.0020 & 0.0100 & 1.00 (Large) \\
\hline
\multicolumn{5}{l}{\footnotesize $p_{\text{Holm}}$: Holm--Bonferroni adjusted (5 comparisons).} \\
\multicolumn{5}{l}{\footnotesize $m_{\text{eff}}$: 9 for Sakamoto, SOLCA, Sequitur (tied on unary); 10 for Bisection, LZD.} \\
\multicolumn{5}{l}{\footnotesize $r = (W^{+} - W^{-})/(W^{+} + W^{-})$. For Sakamoto, SOLCA, Bisection, LZD: $W^{-} = 0 \Rightarrow r = 1.00$.} \\
\multicolumn{5}{l}{\footnotesize For Sequitur: $W^{-} = 2$ (ternary maxASI), $r = (43-2)/45 = 0.91$.} \\
\multicolumn{5}{l}{\footnotesize Mean errors: unweighted averages of ten set-level means.} \\
\end{tabular}
\end{table}

T-NDR outperforms all five baselines with large effect sizes. Against LZD, Sakamoto, Bisection, and SOLCA, all non-tied differences favour T-NDR ($W^{-} = 0$, $r = 1.00$). Against Sequitur ($p_{\text{Holm}} = 0.0117$), the sole set in Sequitur's favour is ternary maxASI (rank 2 among 9 absolute differences), giving $r = 0.91$.

\subsection{Pareto Frontier and Algorithm Selection}

\begin{figure}[htbp]
    \centering
    \includegraphics[width=0.85\linewidth]{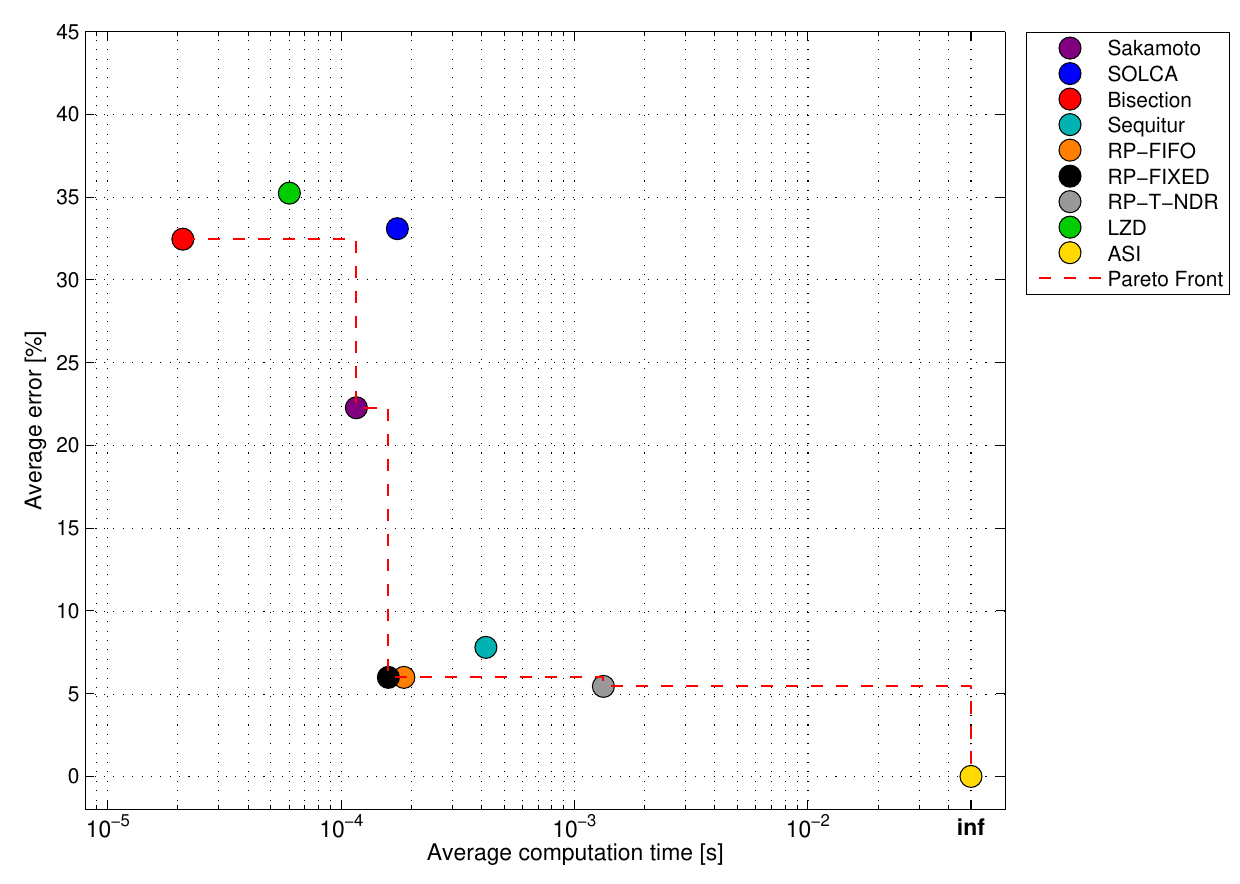}
    \caption{Pareto efficiency: average error vs.\ average computation time across ten synthetic sets. Dashed line: Pareto frontier. ASI finder defines the theoretical optimum.}
    \label{fig:pareto}
\end{figure}

Figure~\ref{fig:pareto} presents the Pareto frontier in average error--average computation time space, aggregated across ten synthetic sets. Per-string times were measured independently; set-level averages were then averaged across all ten sets. The biological set is excluded. The ASI finder is plotted at $t = \infty$, $\text{error} = 0$.

\begin{table}[htbp]
\centering
\begin{threeparttable}
\caption{Algorithm selection guidelines based on Pareto efficiency and string structure.}
\label{tab:final_guidelines}
\small
\begin{tabular}{|l|l|c|c|}
\hline
\textbf{String type} & \textbf{Algorithm} & \textbf{Mean error} & \textbf{$\rho$} \\
\hline
Unary ($|\Sigma|=1$) & Any $\setminus \{\text{Bisection, LZD}\}$ & $\sim$4.2\% & 0.951 \\
Binary maxASI & RP-T-NDR / Sequitur & 4.0--5.3\% & $\geq$0.996 \\
Ternary maxASI & Sequitur / RP-T-NDR & 2.2--2.5\% & $\geq$0.996 \\
Quaternary maxASI & RP-FIXED / T-NDR & 0.00\% & 1.000 \\
\hline
Random entropy (Bin.) & RP-T-NDR / Sequitur & $\sim$15.0\% & 0.90--0.93 \\
Random entropy (Ter.) & Re-Pair family & 8.97\% & 0.925 \\
Random entropy (Quat.) & Re-Pair family & $\sim$5.3\% & $\geq$0.92 \\
\hline
Fixed entropy (Bin.) & RP-T-NDR & 7.52\% & 0.916 \\
Fixed entropy (Ter.) & RP-T-NDR & 3.89\% & 0.987 \\
Fixed entropy (Quat.) & RP-T-NDR & 3.17\% & 0.987 \\
\hline
\end{tabular}
\begin{tablenotes}
\item Note: $\rho$ values are standard Spearman coefficients. For maxASI sets, these are partly inflated; see Table~\ref{tab:partial_spearman} for partial correlations.
\end{tablenotes}
\end{threeparttable}
\end{table}

Re-Pair family dominates across all ten sets, especially T-NDR on fixed-length sets (errors as low as $3.17\%$). Sequitur is the most versatile alternative (not on the Pareto frontier but errors below $10\%$ in 6 of 10 scenarios with minimal computational cost). For sub-microsecond processing, Bisection is the only viable heuristic ($\sim\!2 \times 10^{-5}$\,s), though its error approaches $60\%$ on random binary strings. LZD produces systematic overestimation that renders it unsuitable for AT applications despite reasonable $\rho$ in several sets.

\section{Discussion}\label{sec:Discussion}

\subsection{What the Results Establish}

The maxASI sets are bounded by ASI computation cost: $N \leq 85$ (binary), $N \leq 54$ (ternary), $N \leq 49$ (quaternary). Extending ternary/quaternary to match binary was computationally infeasible; truncating binary/ternary would discard valid results without lowering the quaternary ceiling. These sets therefore cannot be compared across alphabets at face value.

The divergence point (Definition~\ref{def:divergence_point}) and partial Spearman analysis (Table~\ref{tab:partial_spearman}) address this asymmetry directly. A shorter tested range mechanically increases the share of strings below that alphabet's divergence point, inflating rank correlation and exact-match rate regardless of heuristic quality. The ternary maxASI set ($N \leq 54$) is transitional: sufficient to expose divergence onset, but too constrained for conclusive assessment of rank fidelity. 
The quaternary maxASI set ($21 \leq N \leq 49$) is the extreme case: exact matches are consistent with the tested range lying below the quaternary divergence point (i.e., $N^*(\text{T-NDR},4) > 49$), though the sample does not establish this inequality.

The fixed-length sets ($N=30$) remove range asymmetry entirely. There, mean relative error decreases monotonically with $|\Sigma|$ for every algorithm---consistent with the tie-breaking mechanism identified for maxASI sets: greater alphabet diversity reduces digram tie frequency that greedy passes must resolve.

\subsection{Responding to the Introductory Questions}

The first question---whether state-of-the-art compressors offer sufficient heuristic efficacy despite NP-completeness---receives a qualified affirmative, led by Re-Pair T-NDR. This algorithm achieves low, single-digit mean errors across all alphabet sizes (except random binary), establishing superiority over the entire dataset (Table~\ref{tab:wilcoxon}), with the sole exception of ternary maxASI where Sequitur marginally outperforms it.

The practically relevant question is not whether heuristics work, but how far the divergence point $N^*$ extends and how error behaves beyond it. The quaternary maxASI set is uninformative: its formally significant partial correlations reflect only pre-divergence coverage, not genuine discrimination. The fixed-length random sets provide cleanest evidence: T-NDR achieves $\rho \geq 0.916$ (binary), $\rho \geq 0.987$ (ternary), $\rho \geq 0.987$ (quaternary) with no length covariation to remove.

The second question---how $|\Sigma|$ and symbol distribution balance determine accuracy---resolves along two axes. Larger alphabets mitigate greedy tie-breaking, as established above. Entropy acts as a separate driver: Sakamoto degrades toward balanced distributions; SOLCA and Bisection combine high absolute error with variable entropy sensitivity; Re-Pair family and Sequitur remain robust. Combined with the resource-accuracy trade-offs in Fig.~\ref{fig:pareto}, no single heuristic dominates. The practical answer is algorithm-to-use-case matching, formalised in Table~\ref{tab:final_guidelines}.

\subsection{Responding to AT Critics}

Ozelim et al.~\cite{ozelim_assembly_2025-1, ozelim_assembly_2026} claim that \textit{``assembly index is equivalent to LZ-style factorisation''}. This is not supported by structural relationships between AT and the LZ family.

At the theoretical level, a combinatorial separation exists for LZ78-derived schemes (including LZW~\cite{welch_technique_1984} and LZD). Table~\ref{tab:approx_ratio} records the lower bound $\Omega(n^{2/3}/{\log n})$ for the LZ78 family~\cite{charikar_smallest_2005, sakamoto_space-saving_2009}, meaning there are string families on which LZ78-style output exceeds the optimal grammar size $g^*$ by an unbounded factor. This rules out a universal identity between ASI and any LZ78-derived index. The asymptotic nature of this bound does not diminish its force against a claim of strict equivalence; moreover, a single finite counter-example suffices to refute identity, and two suffice to rule out any functional dependence.

The string $0^{64}$ has $\mathrm{ASI}=6$ (by repeated squaring, $0 \to 0^2 \to 0^4 \to \dots \to 0^{64}$, six concatenations) while LZD gives $L_{\mathrm{LZD}}=18$ (Table~\ref{Table:random_b2.pdf}), and the ternary minASI string $(012)^9$ also has $\mathrm{ASI}=6$ while LZD gives $L_{\mathrm{LZD}}=11$ (Table~\ref{Table:random_b3.pdf}). Furthermore, the ASI is direction-independent (invariant under reversal and symbol relabelling), whereas LZ-family algorithms are direction-dependent. This is demonstrated by the mirror pair $w = \texttt{1010201111101112}$ and $w' = \texttt{2111011111020101}$: $\mathrm{ASI}(w)=\mathrm{ASI}(w')=10$, yet $L_{\mathrm{LZD}}(w)=11$ and $L_{\mathrm{LZD}}(w')=12$, showing that LZD output depends on traversal direction. Equal $\mathrm{ASI}$ yields unequal $L_{\mathrm{LZD}}$ (at $\mathrm{ASI}=6$, $L_{\mathrm{LZD}}\in\{11,18\}$), and equal $L_{\mathrm{LZD}}$ yields unequal $\mathrm{ASI}$ (at $L_{\mathrm{LZD}}=11$, $\mathrm{ASI}\in\{6,10\}$), ruling out a functional relationship in either direction.

Ozelim et al.'s claim that \textit{``the index did not need any proof because it is trivially computable''} conflates decidability with tractability. The ASI is computable (exhaustive search exists) but computationally intractable (NP-complete). This is analogous to claiming the Travelling Salesman Problem is trivial because route enumeration exists. Computability establishes decidability; tractability establishes feasibility. The ASI satisfies the former but fails the latter.

Abrahão et al.~\cite{abrahao_assembly_2024} conclude that \textit{``AT and its assembly index represent a considerably constrained version of compression algorithms, and a loose upper bound of $K$''}. The first clause---that ASI is ``a loose upper bound of $K$''---is literally correct: every grammar size upper-bounds Kolmogorov complexity, and loosely so, since CFGs are weaker than Turing machines. This is unremarkable and does not diminish AT. The second clause---that AT is ``a considerably constrained version of compression algorithms''---does not follow from ASI $=$ SLP. If ASI $=$ SLP, then ASI \emph{is} the theoretical optimum of grammar-based compression, not a constrained approximation; no grammar-based method can produce a smaller grammar by definition.
Beyond our analysis, empirical evaluation of chemical samples \cite{tang_distributed_2026} confirms that AT (via the MA index) fundamentally diverges from the Lempel-Ziv (LZW) framework.

\subsection{Beyond SLP: What AT Enables}

Mathematical equivalence ASI $=$ SLP establishes that AT and grammar-based compression optimise the same objective function, but not that they ask identical questions. AT extends beyond the SGP framework in several directions.

First, AT identifies extremal complexity bounds---minASI and maxASI---that remain secondary in pure compression theory. For all unary and minASI binary strings, the lower bound is strictly determined by addition-chain combinatorics: $\text{ASI}(a^N)$ equals the length of the shortest addition chain for $N$. By probing what grammar-based substitution \emph{cannot} compress, AT transforms compression theory from a data-reduction tool into a framework for identifying maximally complex objects.

Second, AT relates ASI to complementary constructs absent from the SGP. In prior work~\cite{bieniawski_assembly_2026}, we established negative correlation between ASI and assembly depth (ASD)---the minimum pathway depth in ASI steps---which is positively correlated with the number of independent assembly steps. We also found a negative correlation between ASI and expected waiting time (EWT). These relationships have no analogue in grammar-based compression, which treats grammar size as an isolated optimisation target.

\section{Conclusions}\label{sec:Conclusions}

The ASI represents the output of an ideal, optimal compressor, defining the theoretical global minimum of the SGP.
The empirical results establish that GBCA provide reliable practical upper bounds on the ASI (where the ASI is computable), with the Re-Pair family --- and T-NDR in particular --- dominating the accuracy end of the trade-off.

The results presented here are bounded by the string lengths tractable for ASI computation and by linear string representations that do not model geometric, energetic, or synthesis-pathway constraints.

The precision of GBCA as ASI estimators established here has direct practical implications: in real-world applications, minimal grammar size differences may cross application-specific thresholds, justifying the tighter bound provided by Re-Pair T-NDR.

{\bf Funding---} This research received no external funding.

{\bf Data Availability Statement---} 
The implementations of the evaluated algorithms were obtained from the following public repositories:

\begin{itemize}
    \item Sakamoto and Bisection: \url{https://github.com/andrewwinslow/grammar-approx}
    \item SOLCA: \url{https://github.com/tkbtkysms/solca}
    \item Sequitur: \url{https://github.com/michaliskambi/grammar-compression}
    \item Re-Pair FIFO: \url{https://github.com/rwanwork/Re-Pair}
    \item LZD: \url{https://github.com/kg86/lzd}
\end{itemize}

Pseudocode for Re-Pair FIXED is given in Algorithm~\ref{alg:fixed_repair}; pseudocode for Re-Pair T-NDR is given in Algorithm~\ref{alg:t_ndr_final_fixed}.

Modified implementations of the evaluated algorithms written in Python are available at: 

\url{https://github.com/wbieniawski-123/ASI-vs-other-algos} (accessed on 23 June 2026)

C++ computational environment for general ASI computation is available at:

\url{https://github.com/szluk/AssemblyTheory} (accessed on 23 June 2026)

{\bf Acknowledgements---}  
The author thanks Szymon Łukaszyk and Piotr Masierak for their clarifications, formal corrections, and improvements.
We acknowledge Polish high-performance computing infrastructure PLGrid for awarding this project access to the LUMI supercomputer, owned by the EuroHPC Joint Undertaking, hosted by CSC (Finland) and the LUMI consortium through the pllatsl02 (19070) grant.

{\bf Conflicts of Interest---} 
No conflict of interest is declared by the author.

\appendix

\section*{Abbreviations}
The following abbreviations are used in this manuscript:

\begin{tabular}{ll}
AT                                       & assembly theory;\\
ASI                                      & assembly index;\\
$N$                                      & length of a string;\\
$|\Sigma|$                               & size of an alphabet;\\
SGP                                      & smallest grammar problem;\\
SLP                                      & straight-line program;\\
CFG                                      & context free grammar;\\
GBCA                                     & grammar-based compression algorithms;\\
LCA                                      & lowest common ancestor;\\
$w$                                      & string;\\
$H(w)$                                   & Shannon entropy;\\
\end{tabular}

\section*{Algorithms and Tables}

The following algorithms and tables are referenced in this manuscript:

\begin{algorithm}[H]
\setstretch{1.2}
\DontPrintSemicolon
\caption{Re-Pair FIXED}
\label{alg:fixed_repair}

\KwIn{string $w$}
\KwOut{Grammar size $L$, Rule set $\mathcal{R}$}

\BlankLine
$w_{list} \leftarrow \text{list of symbols from } w$ \;
$\mathcal{R} \leftarrow \emptyset$ \;
$r \leftarrow 0$ \;

\BlankLine
\tcp{Iterative greedy reduction}
\While{True}{
    $Pairs \leftarrow \text{count\_adjacent\_pairs}(w_{list})$ \;
    
    \If{$Pairs = \emptyset$ \textbf{or} $\max \{ Pairs[p] : p \in Pairs \} < 2$}{
        \textbf{break} \tcp*{Stop when no pair repeats}
    }

    \BlankLine
    \tcp{Step 1: Find maximum frequency}
    $f_{max} \leftarrow \max \{ Pairs[p] : p \in Pairs \}$ \;
    $Candidates \leftarrow \{ p \in Pairs : Pairs[p] = f_{max} \}$ \;

    \BlankLine
    \tcp{Step 2: Deterministic Tie-breaking (Lexicographical)}
    Sort $Candidates$ lexicographically \;
    $pair_{best} \leftarrow Candidates[0]$ \;

    \BlankLine
    \tcp{Step 3: Rule Creation and Substitution}
    $\mathcal{R} \leftarrow \mathcal{R} \cup \{ R_r \rightarrow pair_{best} \}$ \;
    $w_{list} \leftarrow \text{replace\_non\_overlapping}(w_{list}, pair_{best}, R_r)$ \;
    $r \leftarrow r + 1$ \;
}

\BlankLine
\tcp{Final Complexity Estimation (L)}
$N_{rules} \leftarrow \text{length}(\mathcal{R})$ \;
$N_{w} \leftarrow \max(0, \text{length}(w_{list}) - 1)$ \;
$L \leftarrow N_{rules} + N_{w}$ \;

\Return $L, \mathcal{R}$ \;
\end{algorithm}

\begin{algorithm}[H]
\setstretch{1.2}
\DontPrintSemicolon
\caption{Re-Pair T-NDR}
\label{alg:t_ndr_final_fixed}

\KwIn{string $w$, current rule count $r$}
\KwOut{Minimal grammar size $L_{min}$}

\BlankLine
\tcp{Global shared state for branch-and-bound pruning}
\textbf{Global Variables:} $G_{best} \leftarrow \infty$, $Cache \leftarrow \emptyset$ \;

\SetKwFunction{FTNDR}{T-NDR-Recursive}
\SetKwProg{Fn}{Function}{:}{}

\Fn{\FTNDR{$w, r$}}{
    $N \leftarrow \text{length}(w)$ \;
    $S_{curr} \leftarrow r + (N - 1)$ \tcp*{Current grammar size if we stop here}
    $S_{lower} \leftarrow r + \lceil \log_2 N \rceil$ \tcp*{True mathematical lower bound}

    \BlankLine
    \tcp{Step 1: Normalisation and Memoisation}
    $w_{norm} \leftarrow$ \textit{canonical\_form}($w$) \;
    \If{$w_{norm} \in Cache$}{
        \Return $r + Cache[w_{norm}]$ \;
    }

    \tcp{Early pruning using the true lower bound}
    \If{$S_{lower} \geq G_{best}$}{
        \Return $\infty$ \;
    }

    \BlankLine
    \tcp{Step 2: Base Case Analysis}
    $Pairs \leftarrow$ \textit{count\_adjacent\_pairs}($w$) \;
    \If{$Pairs = \emptyset$ \textbf{or} $\max \{ Pairs[p] : p \in Pairs \} < 2$}{
        $G_{best} \leftarrow \min(G_{best}, S_{curr})$ \;
        \Return $S_{curr}$ \;
    }

    \BlankLine
    \tcp{Step 3: Branching with Dynamic Pruning}
    $min\_L \leftarrow \infty$ \;
    $f_{max} \leftarrow \max \{ Pairs[p] : p \in Pairs \}$ \;
    $Candidates \leftarrow \{p \in Pairs : Pairs[p] = f_{max} \}$ \;
    
    \For{\textbf{each} $pair \in Candidates$}{
        $w_{next} \leftarrow$ \textit{replace\_non\_overlapping}($w, pair, R_r$) \;
        $res \leftarrow$ \FTNDR{$w_{next}, r + 1$} \;
        
        \BlankLine
        \tcp{Update both local minimum and global bound immediately}
        $min\_L \leftarrow \min(min\_L, res)$ \;
        $G_{best} \leftarrow \min(G_{best}, res)$ \;
    }

    \BlankLine
    \tcp{Step 4: Cache Residual Cost}
    $Cache[w_{norm}] \leftarrow min\_L - r$ \;
    \Return $min\_L$ \;
}
\end{algorithm}

\begin{table}[H]
  \centering
  \includegraphics[width=\linewidth]{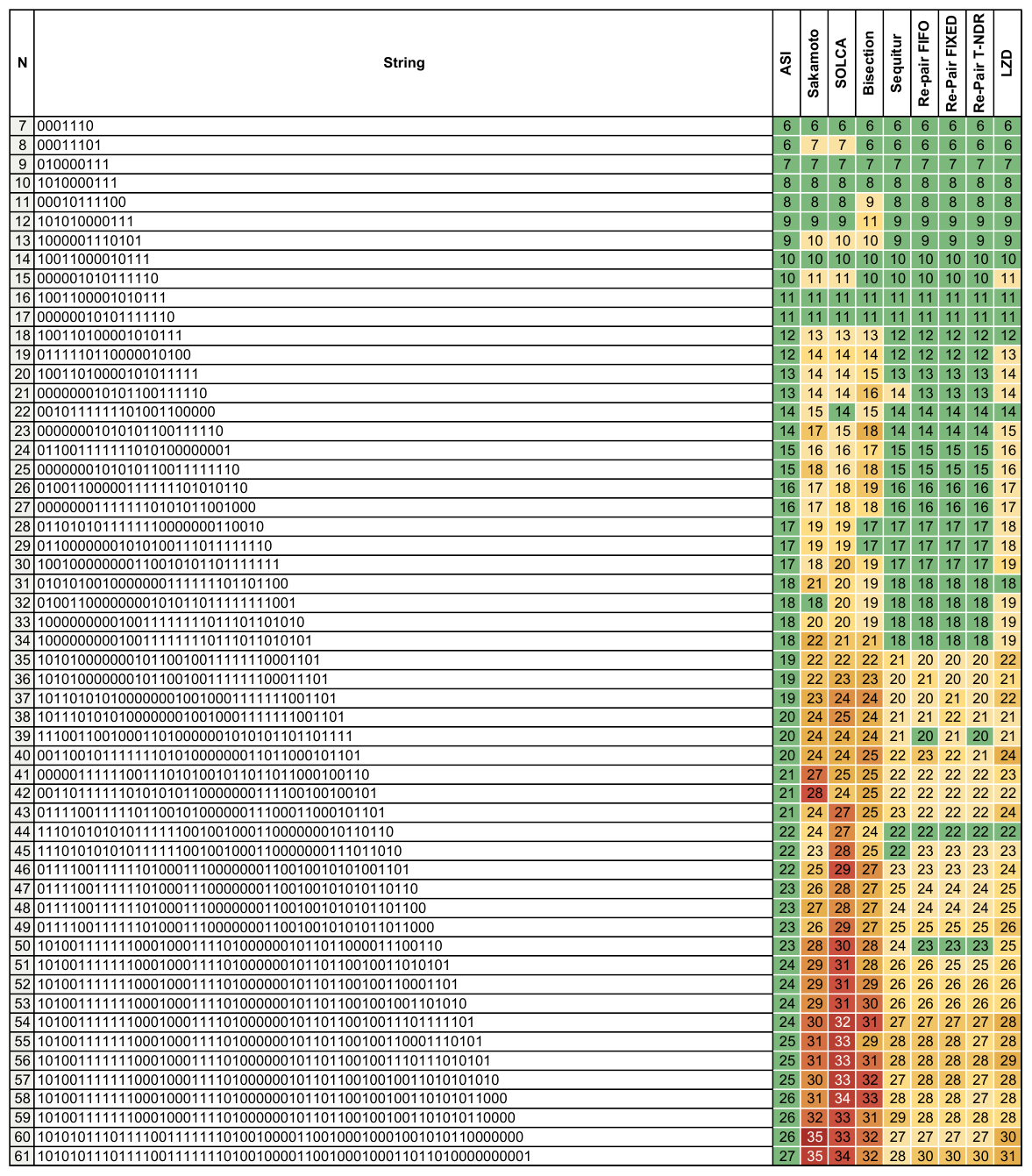}
  \captionsetup{name=Table}
  \caption{Binary maxASI strings ($\Sigma = 2$) for $7 \le N \le 60$.}
  \label{Table:maxbinarya}
\end{table}

\begin{table}[H]
  \centering
  \includegraphics[width=\linewidth]{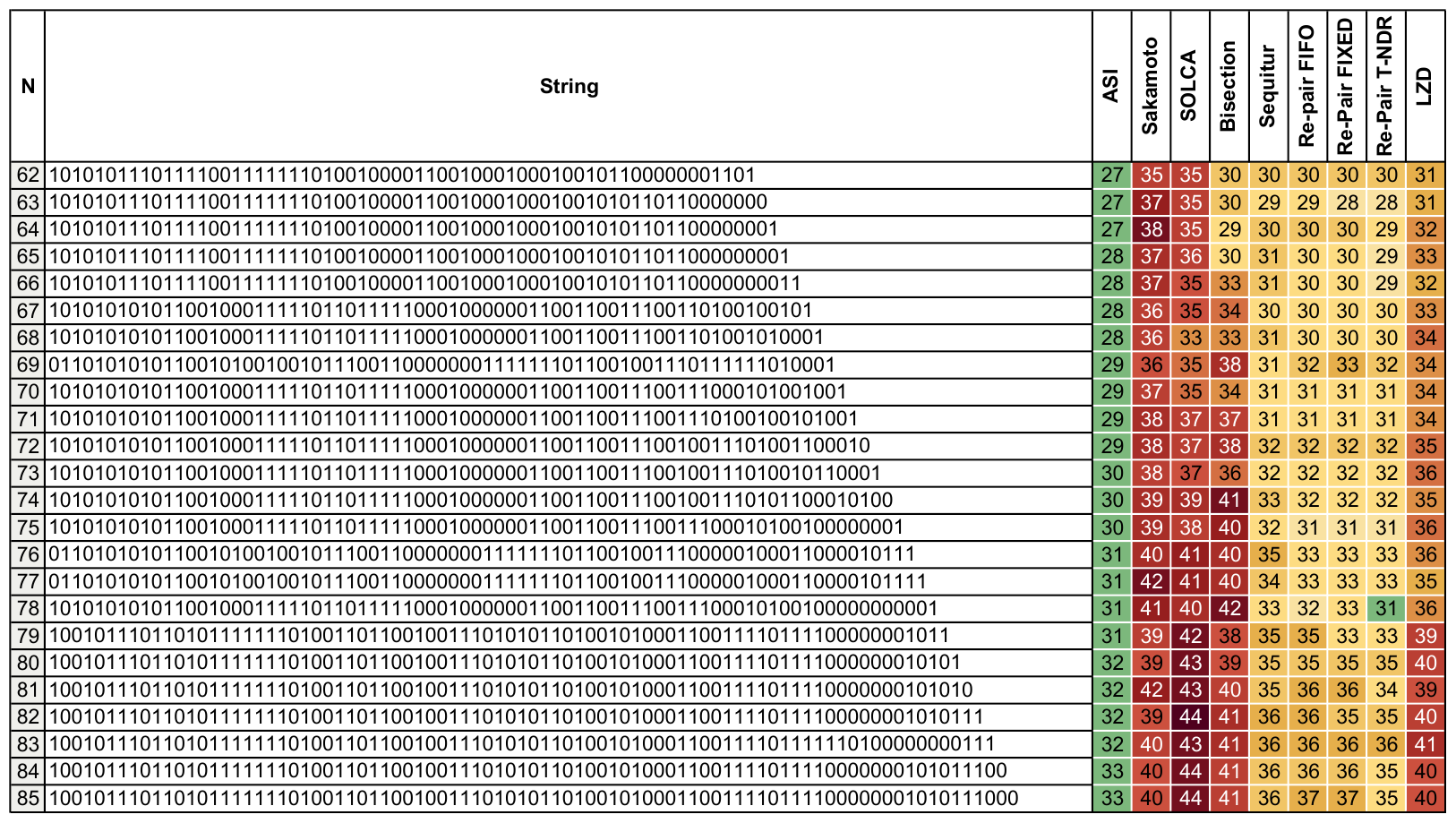}
  \captionsetup{name=Table}
  \caption{Binary maxASI strings ($\Sigma = 2$) for $61 \le N \le 85$.}
  \label{Table:maxbinaryb}
\end{table}

\begin{table}[H]
  \centering
  \includegraphics[width=\linewidth]{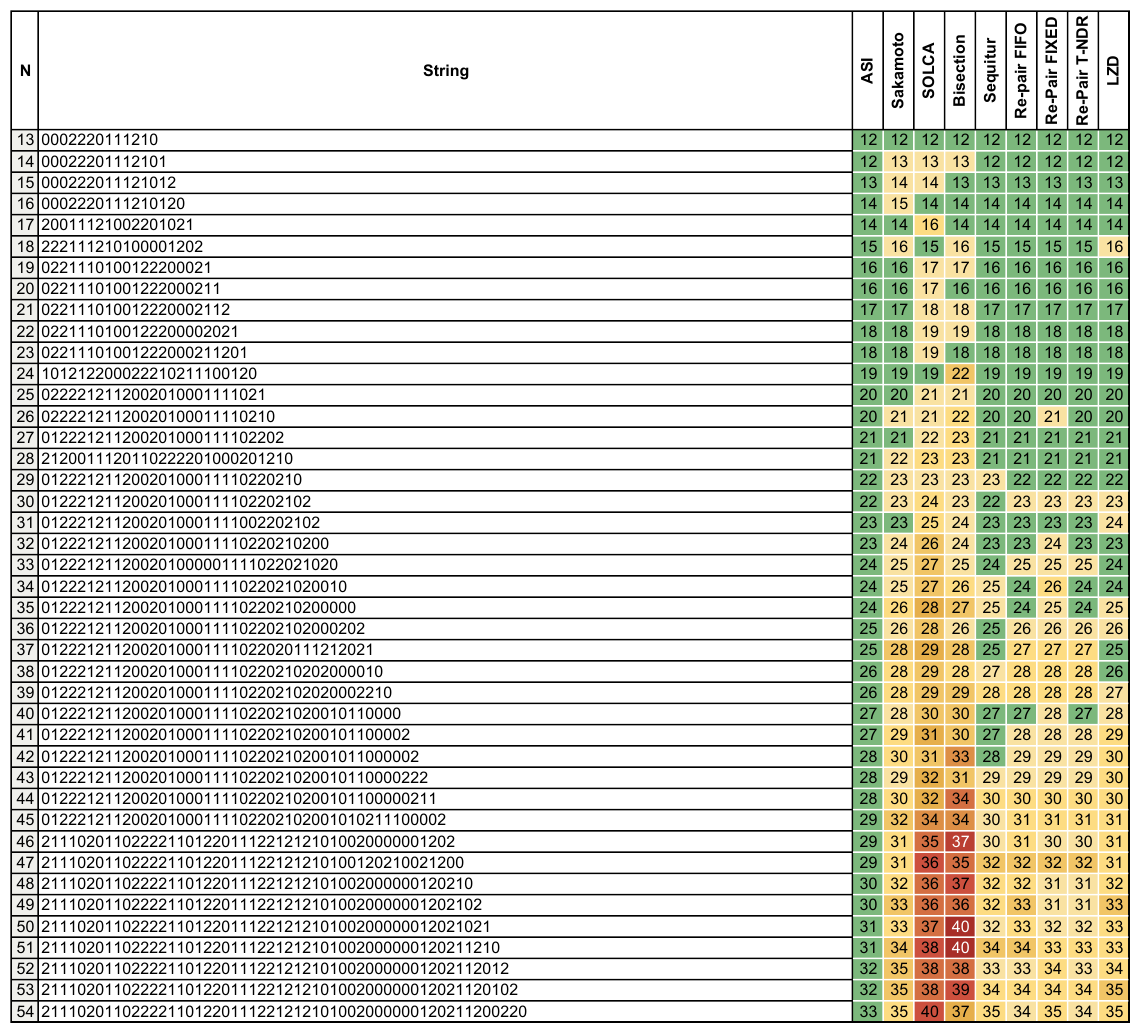}
  \captionsetup{name=Table}
  \caption{Ternary maxASI strings ($\Sigma = 3$) for $13 \le N \le 54$.}
  \label{Table:maxternary}
\end{table}

\begin{table}[H]
  \centering
  \includegraphics[width=\linewidth]{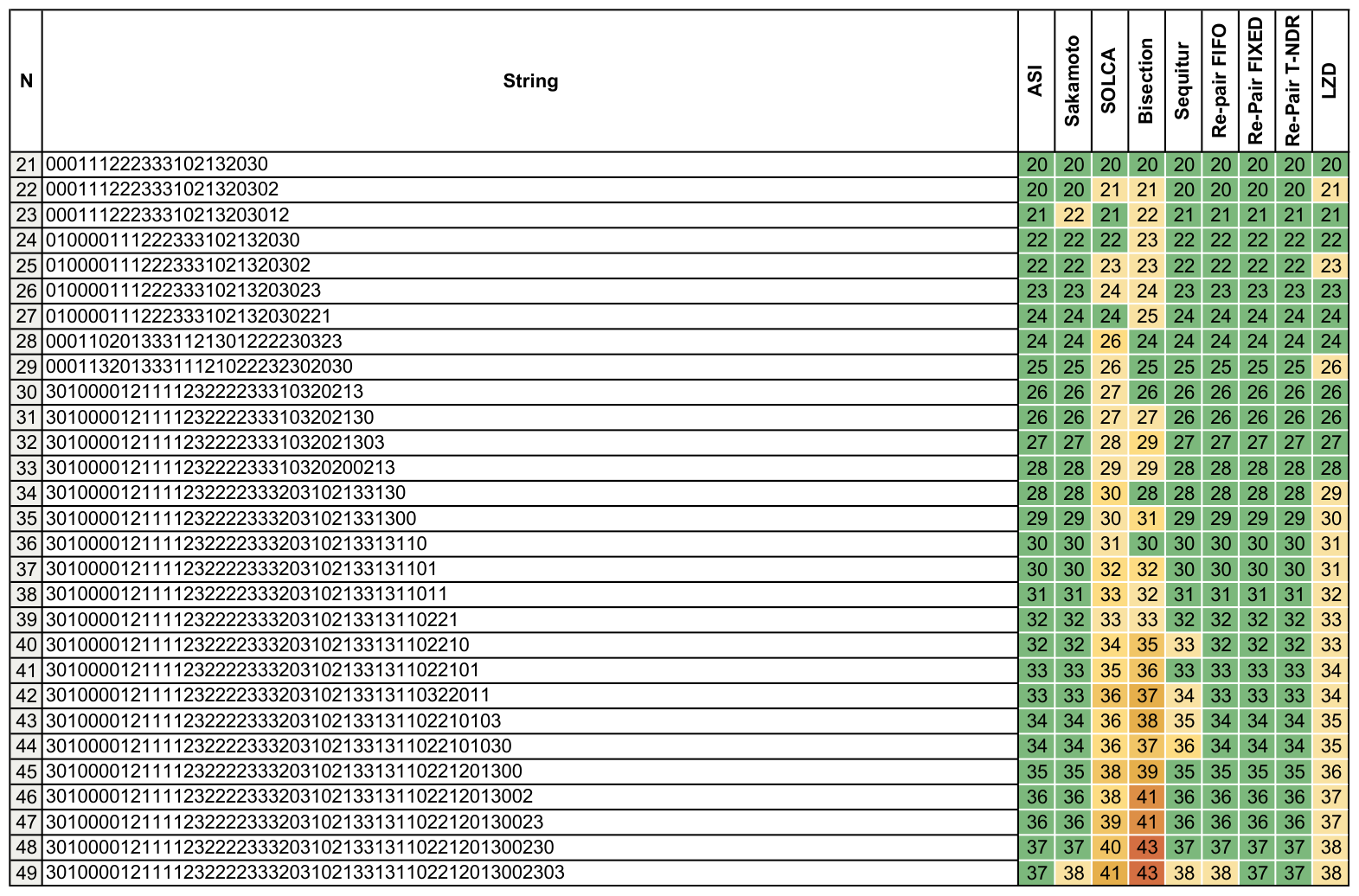}
  \captionsetup{name=Table}
  \caption{Quaternary maxASI strings ($\Sigma = 4$) for $21 \le N \le 49$.}
  \label{Table:maxquaternary}
\end{table}

\begin{table}[H]
  \centering
  \includegraphics[width=\linewidth]{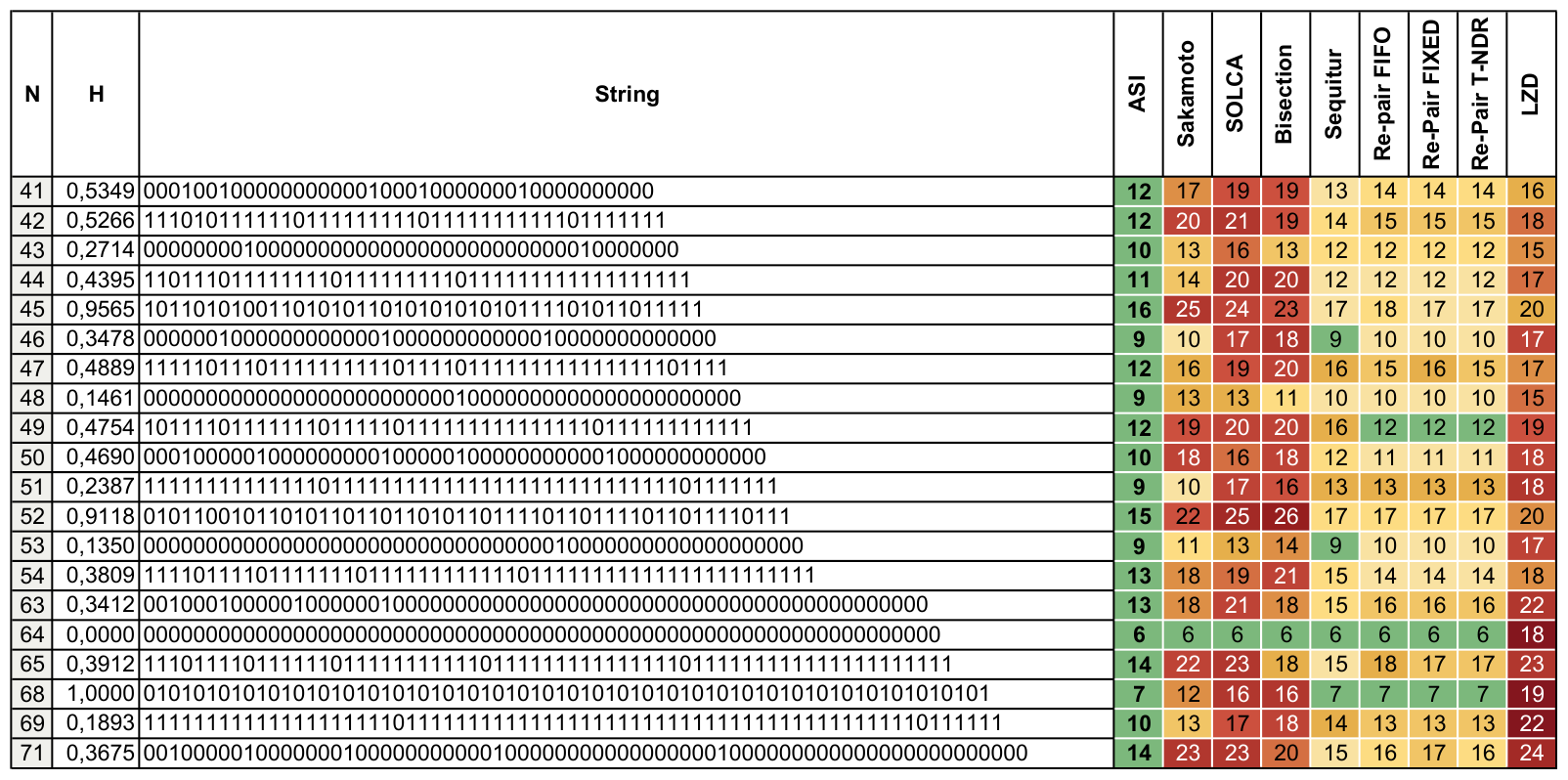}
  \captionsetup{name=Table}
  \caption{Random binary strings.}
  \label{Table:random_b2.pdf}
\end{table}

\begin{table}[H]
  \centering
  \includegraphics[width=\linewidth]{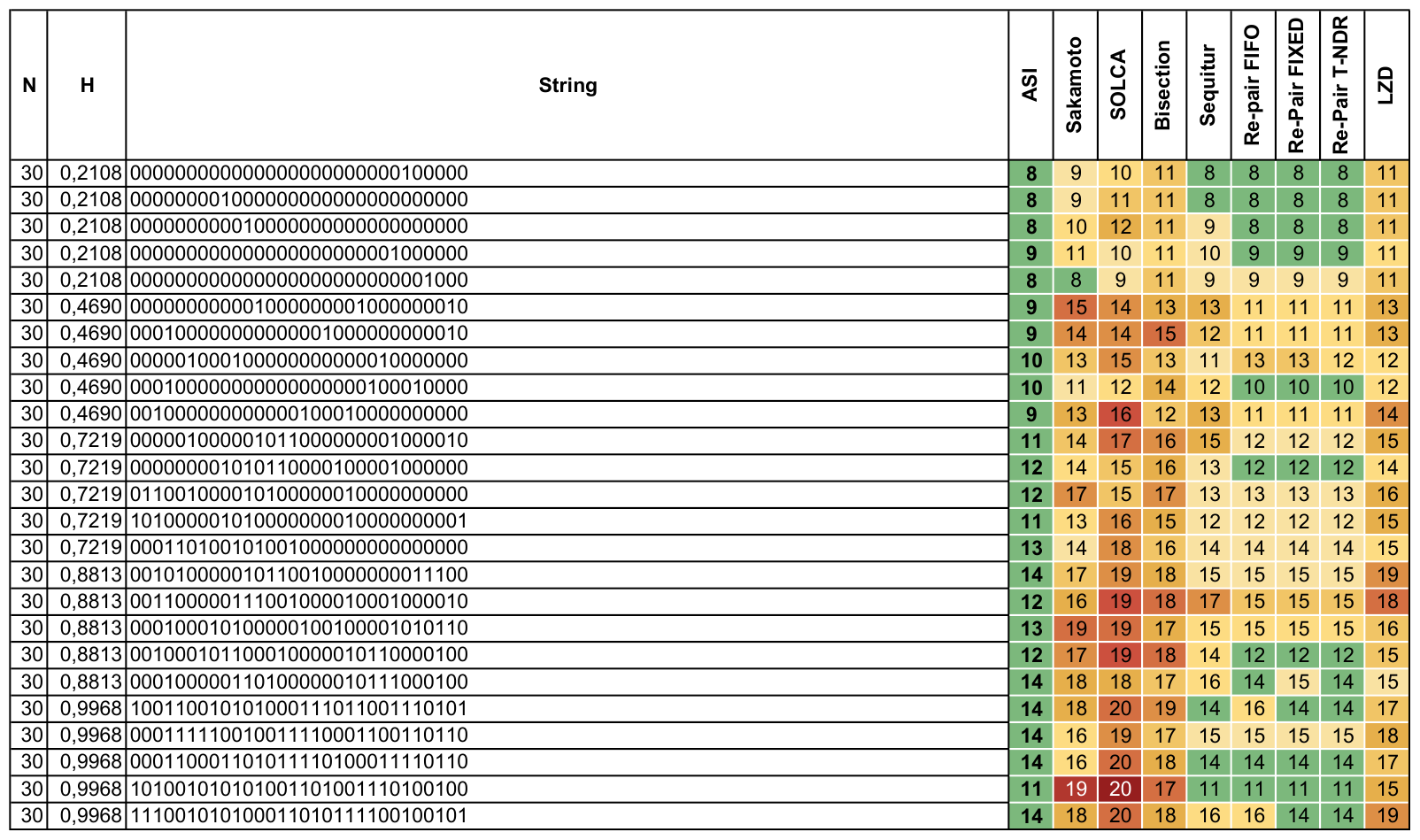}
  \captionsetup{name=Table}
  \caption{Random fixed-length binary strings.}
  \label{Table:fixed_b2.pdf}
\end{table}

\begin{table}[H]
  \centering
  \includegraphics[width=\linewidth]{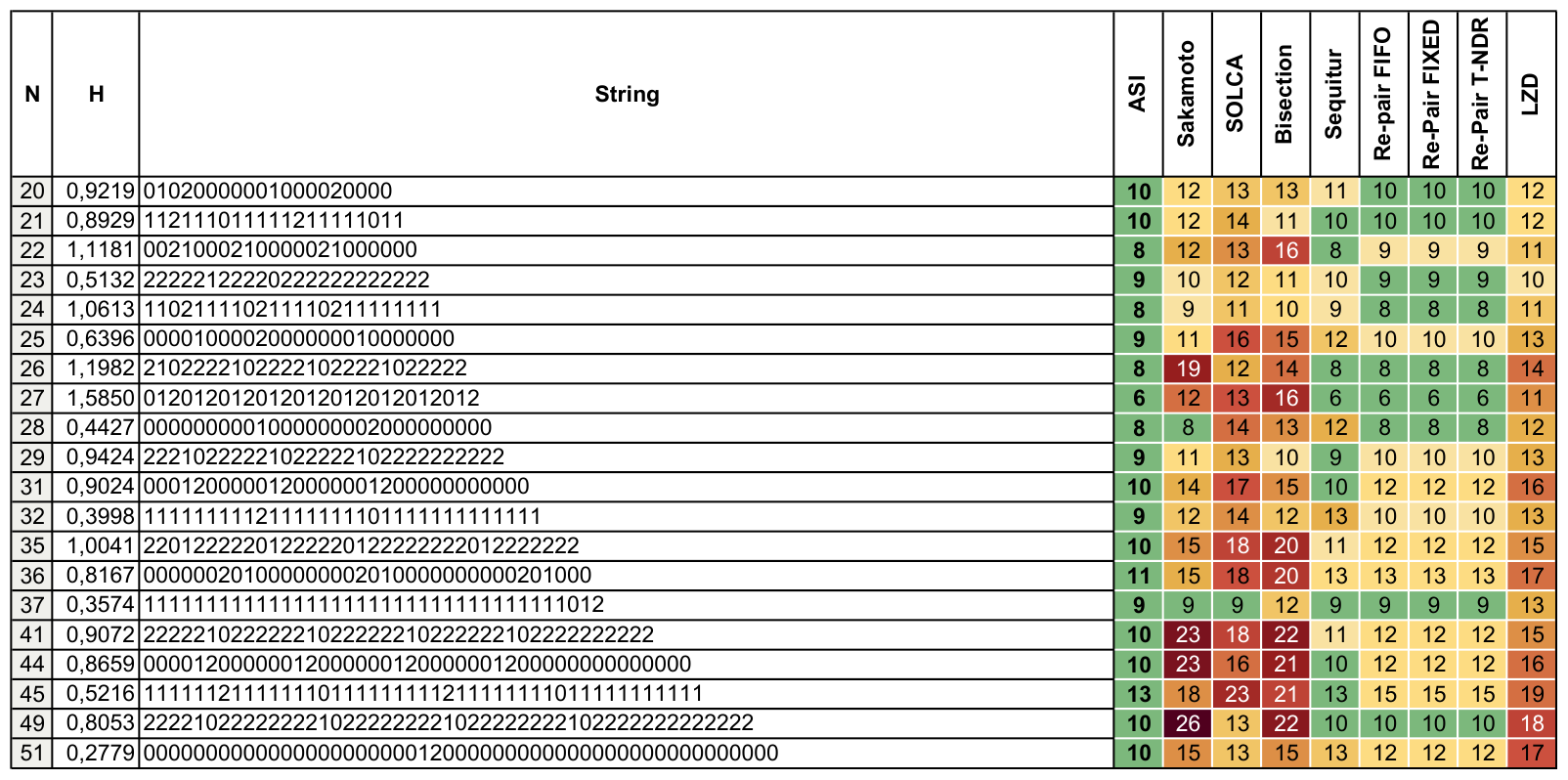}
  \captionsetup{name=Table}
  \caption{Random ternary strings.}
  \label{Table:random_b3.pdf}
\end{table}

\begin{table}[H]
  \centering
  \includegraphics[width=\linewidth]{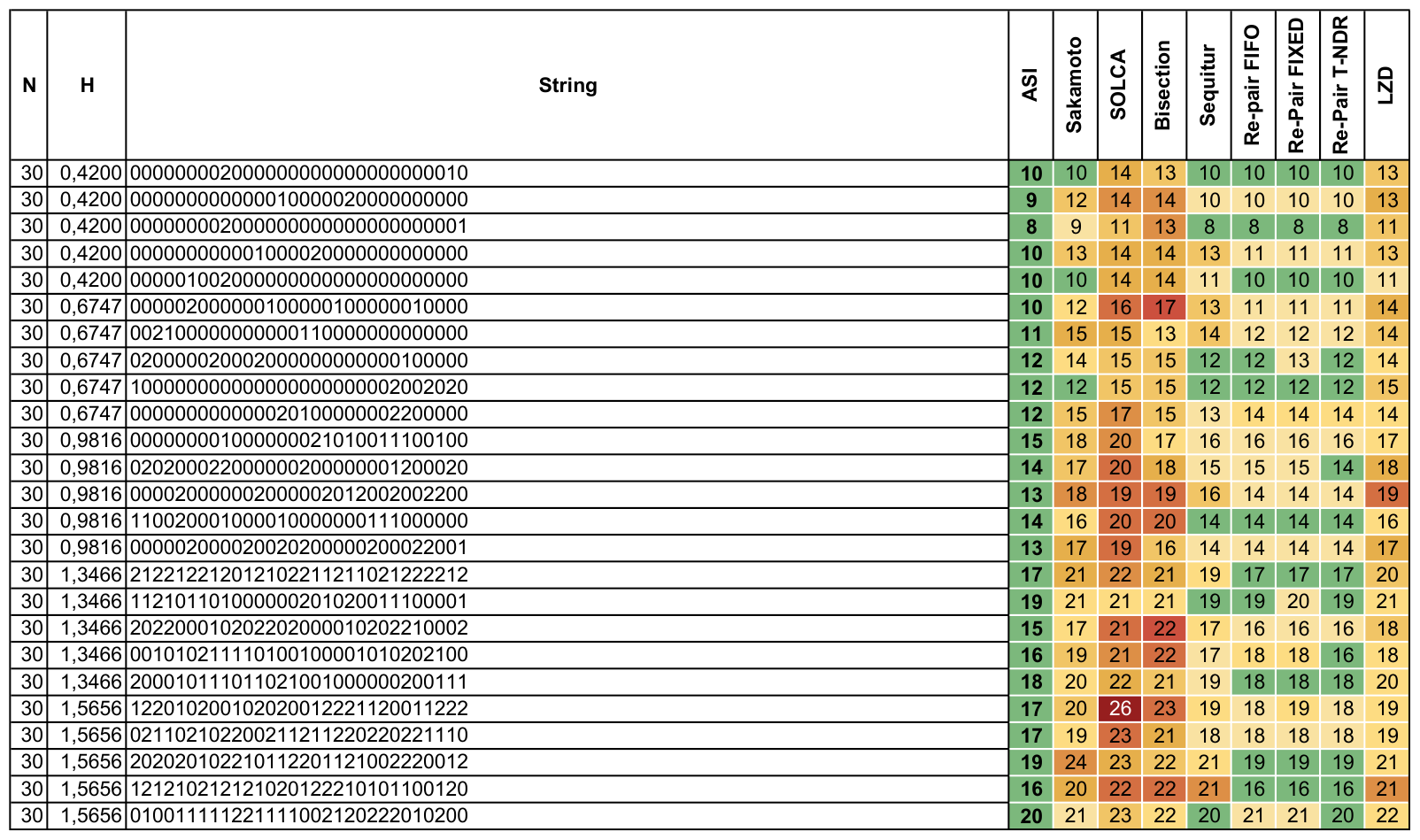}
  \captionsetup{name=Table}
  \caption{Random fixed-length ternary strings.}
  \label{Table:fixed_b3.pdf}
\end{table}

\begin{table}[H]
  \centering
  \includegraphics[width=\linewidth]{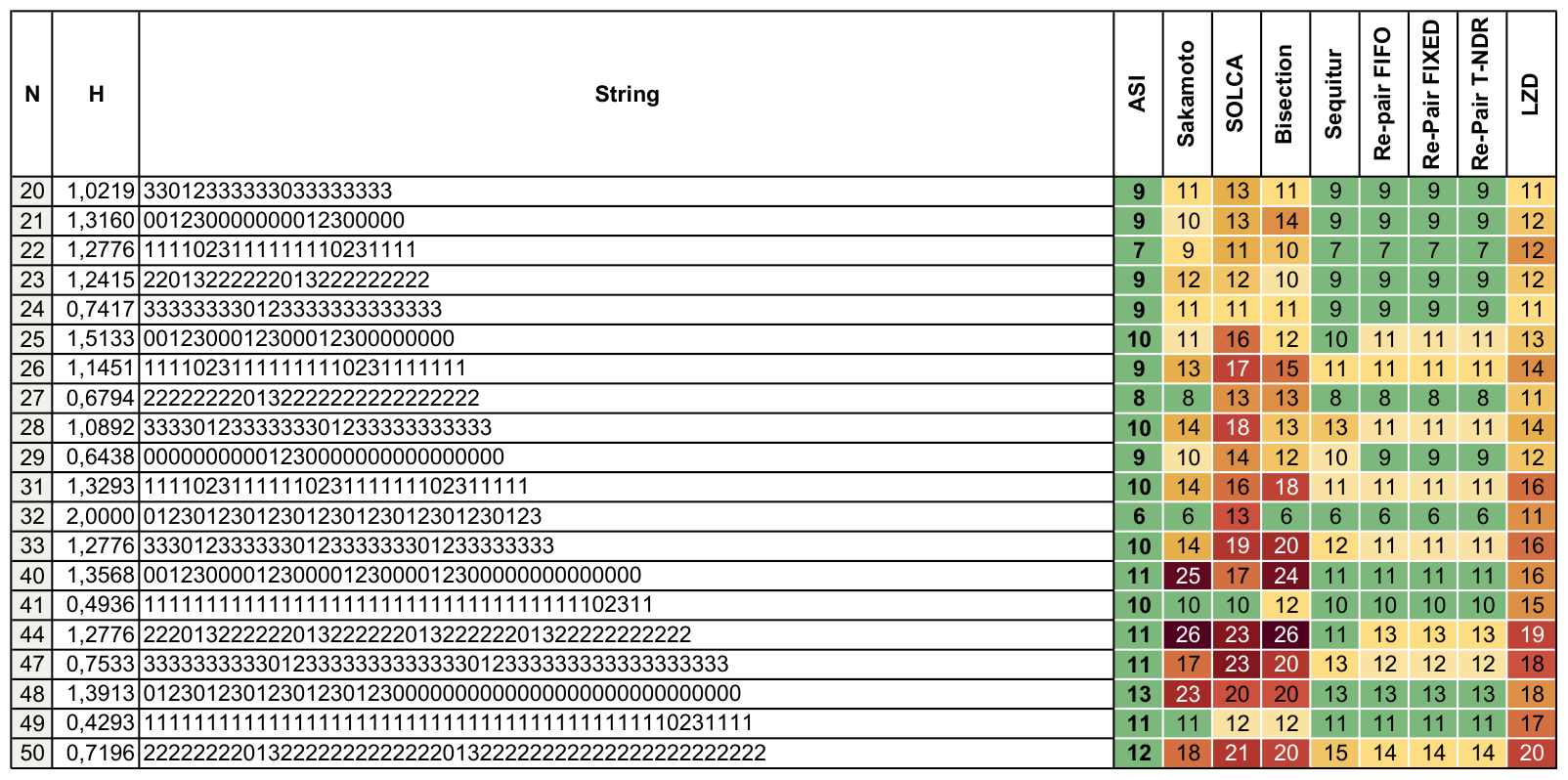}
  \captionsetup{name=Table}
  \caption{Random quaternary strings.}
  \label{Table:random_b4.pdf}
\end{table}

\begin{table}[H]
  \centering
  \includegraphics[width=\linewidth]{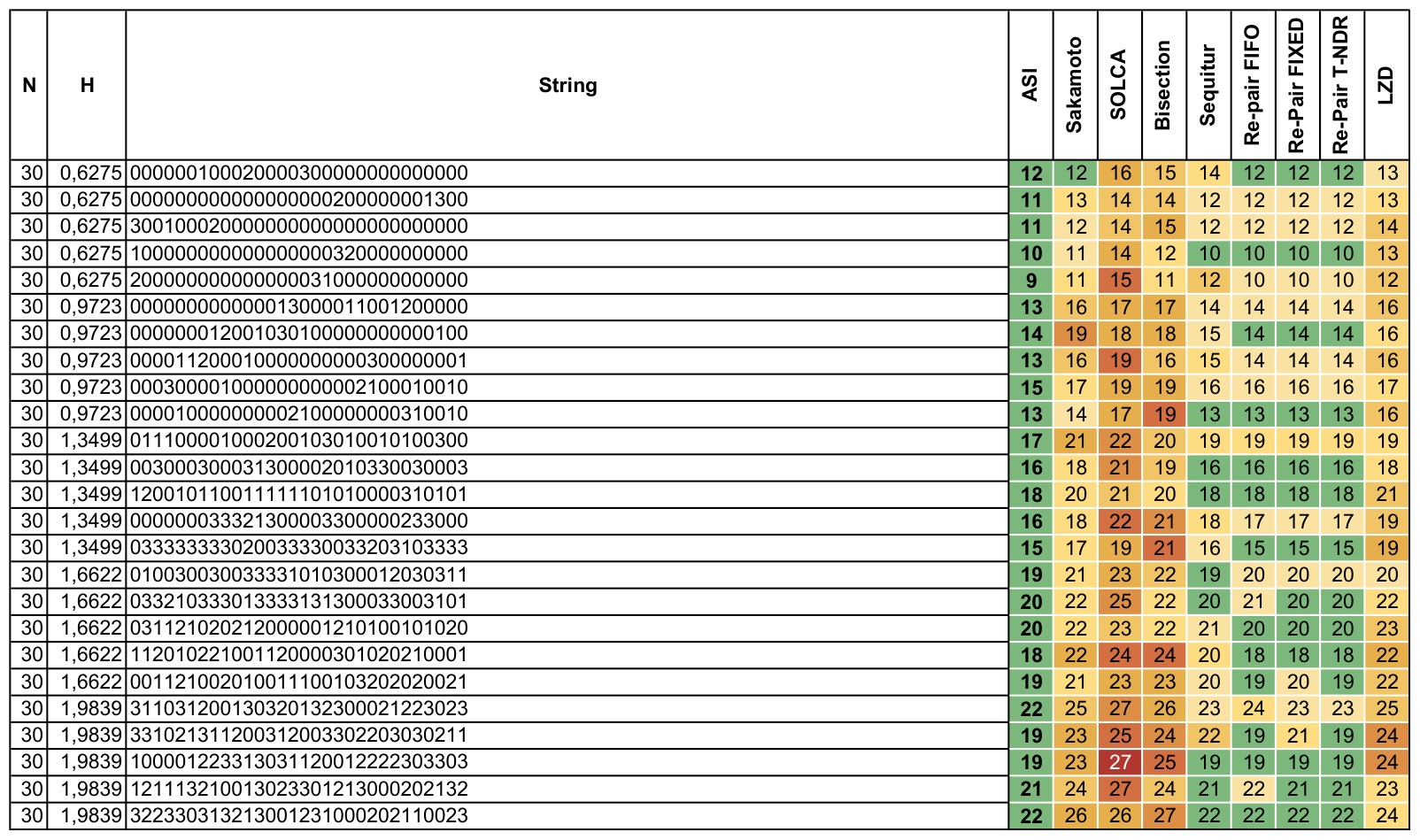}
  \captionsetup{name=Table}
  \caption{Random fixed-length quaternary strings.}
  \label{Table:fixed_b4.pdf}
\end{table}

\begin{table}[H]
\centering
\caption{Compression results for biological sequences: Proteins ($|\Sigma| = 18-20$) and Nucleic Acids ($|\Sigma| = 4$). Values marked 'n.a.' exceeded computational limits for the ASI calculation.}
\label{tab:biological_results_all}
\footnotesize 
\begin{tabular}{@{}lcccccccccccc@{}}
\toprule
\textbf{Accession ID} & $N$ & $|\Sigma|$ & \textbf{ASI} & \textbf{Sak.} & \textbf{SOL.} & \textbf{Bis.} & \textbf{Seq.} & \textbf{FIFO} & \textbf{FIX.} & \textbf{T-NDR} & \textbf{LZD} \\ \midrule
\multicolumn{12}{l}{\textit{Polypeptide Chains (Proteins)}} \\
1UBQ & 76 & 18 & 64 & 64 & 71 & 75 & 67 & 64 & 64 & 64 & 68 \\
1MUI & 99 & 20 & n.a. & 86 & 94 & 97 & 85 & 85 & 85 & 85 & 89 \\
1BNR & 110 & 18 & n.a. & 97 & 102 & 105 & 96 & 96 & 96 & 96 & 99 \\
2RNF & 120 & 20 & n.a. & 106 & 114 & 112 & 103 & 103 & 104 & 103 & 111 \\
1BP2 & 123 & 20 & n.a. & 107 & 116 & 115 & 107 & 106 & 106 & 105 & 109 \\
3EV6 & 124 & 19 & n.a. & 114 & 120 & 119 & 114 & 114 & 113 & 113 & 116 \\
3CHY & 128 & 18 & n.a. & 111 & 123 & 118 & 109 & 109 & 110 & 108 & 114 \\
4XJD & 129 & 20 & n.a. & 117 & 121 & 122 & 116 & 115 & 117 & 115 & 120 \\
2LIS & 136 & 19 & n.a. & 119 & 126 & 129 & 119 & 118 & 118 & 118 & 125 \\
2DN1 ($\alpha$) & 141 & 19 & n.a. & 120 & 132 & 131 & 120 & 119 & 119 & 119 & 125 \\
2DN1 ($\beta$) & 146 & 19 & n.a. & 131 & 138 & 136 & 125 & 126 & 125 & 125 & 129 \\
1MBA & 147 & 19 & n.a. & 126 & 135 & 133 & 125 & 125 & 124 & 123 & 127 \\
4L79 & 149 & 18 & n.a. & 123 & 133 & 139 & 123 & 121 & 121 & 120 & 125 \\
8T9Q & 149 & 18 & n.a. & 125 & 134 & 136 & 123 & 123 & 122 & 122 & 129 \\
8VQP & 154 & 19 & n.a. & 132 & 143 & 142 & 129 & 130 & 130 & 130 & 133 \\
1CPC (A,C) & 162 & 20 & n.a. & 137 & 151 & 154 & 135 & 136 & 136 & 135 & 143 \\
2LZM & 164 & 20 & n.a. & 142 & 158 & 155 & 143 & 140 & 141 & 138 & 144 \\
1CPC (B,D) & 172 & 18 & n.a. & 147 & 158 & 161 & 142 & 143 & 144 & 142 & 151 \\
8T4Q & 172 & 20 & n.a. & 152 & 167 & 167 & 147 & 147 & 147 & 147 & 154 \\
3GBN & 179 & 20 & n.a. & 158 & 170 & 174 & 156 & 154 & 154 & 153 & 164 \\ \midrule
\multicolumn{12}{l}{\textit{Nucleic Acid sequences (Genomes/RNA)}} \\
NR\_029492 & 71 & 4 & 36 & 43 & 49 & 49 & 42 & 40 & 40 & 37 & 44 \\
NR\_029820 & 80 & 4 & n.a. & 49 & 57 & 61 & 45 & 44 & 45 & 44 & 50 \\
NR\_031554 & 86 & 4 & n.a. & 55 & 61 & 67 & 51 & 50 & 51 & 50 & 58 \\
NR\_003686 & 96 & 4 & n.a. & 57 & 65 & 61 & 56 & 53 & 53 & 53 & 58 \\
NR\_029778 & 97 & 4 & n.a. & 59 & 68 & 68 & 55 & 57 & 57 & 55 & 63 \\
NR\_029700 & 99 & 4 & n.a. & 57 & 70 & 73 & 54 & 55 & 55 & 54 & 57 \\
NR\_029701 & 99 & 4 & n.a. & 64 & 67 & 73 & 56 & 55 & 54 & 53 & 62 \\
NR\_031970 & 100 & 4 & n.a. & 65 & 67 & 75 & 60 & 58 & 58 & 57 & 61 \\
NR\_029676 & 102 & 4 & n.a. & 59 & 70 & 73 & 61 & 58 & 58 & 58 & 65 \\
NR\_029610 & 110 & 4 & n.a. & 66 & 73 & 75 & 63 & 59 & 61 & 59 & 65 \\
NR\_029609 & 110 & 4 & n.a. & 70 & 73 & 75 & 63 & 62 & 64 & 62 & 68 \\
NR\_033048 & 110 & 4 & n.a. & 66 & 74 & 76 & 65 & 65 & 64 & 63 & 68 \\
NR\_029856 & 111 & 4 & n.a. & 67 & 75 & 77 & 64 & 65 & 64 & 61 & 71 \\
NR\_031615 & 112 & 4 & n.a. & 68 & 75 & 74 & 59 & 59 & 58 & 58 & 71 \\
M10816 & 117 & 4 & n.a. & 71 & 78 & 73 & 67 & 65 & 67 & 65 & 73 \\
NR\_024244 & 119 & 4 & n.a. & 65 & 77 & 73 & 67 & 64 & 64 & 64 & 74 \\
V00590 & 161 & 4 & n.a. & 92 & 103 & 102 & 90 & 88 & 88 & 88 & 92 \\
NR\_004430 & 164 & 4 & n.a. & 94 & 104 & 108 & 91 & 89 & 90 & 89 & 93 \\
X01042 & 166 & 4 & n.a. & 95 & 105 & 103 & 91 & 86 & 87 & 86 & 92 \\
NR\_002716 & 187 & 4 & n.a. & 105 & 116 & 120 & 102 & 103 & 105 & 101 & 111 \\ \bottomrule
\end{tabular}
\end{table}

% The \nocite command causes all entries in a bibliography to be printed out whether or not they are actually referenced in the text. This is appropriate for the sample file to show the different styles of references, but authors most likely will not want to use it.
%\nocite{*}

%\bibliographystyle{vancouver}
\bibliography{apssamp}

\end{document}